\documentclass[conference,10pt,compsoc]{IEEEtran}
\usepackage[a4paper,scale={0.85,0.85}]{geometry}

\usepackage[utf8]{inputenc}
\usepackage[american]{babel}  
\usepackage{amsmath,amssymb}
\usepackage{amsthm}
\usepackage{mathtools}

\usepackage{graphics}

\usepackage{algorithm}
\usepackage[noend]{algpseudocode}

\floatname{algorithm}{}

\newcommand{\hyp}{\leftarrow}

\usepackage{tikz}
\usetikzlibrary{arrows}
\newtheorem{lemma}{Lemma}
\newtheorem{definition}{Definition}
\newtheorem{theorem}{Theorem}
\newtheorem{remark}{Remark}
\newtheorem{claim}{Claim}
\begin{document}
\title{
  Descriptive complexity for minimal time of cellular automata}

\author{\IEEEauthorblockN{Étienne Grandjean}
  \IEEEauthorblockA{Normandie Univ, UNICAEN, ENSICAEN, CNRS,\\
    GREYC, 14000 CAEN, France\\
  Email: etienne.grandjean@unicaen.fr }
\and
\IEEEauthorblockN{Théo Grente}
\IEEEauthorblockA{Normandie Univ, UNICAEN, ENSICAEN, CNRS, \\
  GREYC, 14000 CAEN, France\\
Email: theo.grente@unicaen.fr}}

\maketitle

\begin{abstract}
Descriptive complexity may be useful to design programs in a natural declarative way. This is important for parallel computation models such as cellular automata, because designing parallel programs is considered difficult. 
Our paper establishes logical characterizations of the three classical complexity classes that model minimal time, called real-time, of one-dimensional cellular automata according to their canonical variants. 
Our logics are natural restrictions of the existential second-order Horn logic. They correspond to the three ways of deciding a language on a square grid circuit of side $n$ according to the three canonical placements of an input word of length $n$ on the grid.
Our key tool is a normalization method that transforms a formula into an equivalent formula that literally mimics a grid circuit.

\end{abstract}

\section{Introduction} 

\subsection{Descriptive complexity and programming}

\noindent
There are two criteria of interest of a complexity class: it contains a number of ``natural'' problems that are \emph{complete} in the class; it has \emph{machine-independent} ``natural'' characterizations, usually in \emph{logic}, i.e., in so-called \emph{descriptive complexity}. The most famous example is Fagin's Theorem~\cite{Fagin74SIAM, Libkin04FMT}, which characterizes $\mathtt{NP}$ as the class of problems definable in \emph{existential second-order logic} ($\mathtt{ESO}$). Similarly, Immerman-Vardi’s Theorem~\cite{Libkin04FMT,Immerman99DesComp} and Grädel’s Theorem~\cite{Gradel92TCS,Gradel07EATCS} characterize the class $\mathtt{P}$ by \emph{first-order logic plus least fixed-point}, and \emph{second-order logic restricted to Horn formulas}, respectively.

Another interest of descriptive complexity is that it allows to automatically derive from a logical description of a problem a program that solves it. This is particularly interesting for the design of \emph{parallel programs} that is considered a difficult task. Typically, a number of algorithmic problems (product of integers, product of matrices, sorting, etc.) are computable in linear time on \emph{cellular automata} (CA), a local and massively parallel model. For each such problem, the literature presents an ``ad hoc’’ parallel and local algorithmic strategy and gives the program of the final CA in an informal way~\cite{Fischer65JACM,Delorme02SCA}. However, the problems in concern can be defined inductively in a natural way. For instance, the product of two integers in binary notation is simply defined by the classical Horner's method and one may hope to directly derive a parallel program from such an inductive process.

\subsection{Descriptive complexity 
and linear time on cellular automata}

\noindent
The present paper is in some sense the sequel of a recent paper~\cite{BacqueyGO17ICALP} 
(see also~\cite{GrandjeanO16}). First, \cite{BacqueyGO17ICALP} observes that the inductive processes defining the problems in concern (product of integers, product of matrices, sorting, etc.) are ``local'' and are naturally formalized by \emph{Horn formulas}, that is by conjunctions of first-order Horn clauses. Therefore, the computation is nothing else than the classical resolution method on Horn clauses, as in Prolog and Datalog~\cite{Libkin04FMT,Gradel07EATCS},~\cite{AbiteboulHV95}. Moreover, on every concrete problem defined by a Horn formula with $d + 1$ first-order variables, this inductive computation by Horn rules can be geometrically modeled as the displacement of a d-dimensional hyperplan along some fixed line in a space of dimension $d + 1$. 
To capture these inductive behaviors, \cite{BacqueyGO17ICALP} defines a logic denoted $\mathtt{monot}$-$\mathtt{ESO}$-$\mathtt{HORN}^d(\forall^{d+1}, \mathtt{arity}^{d+1})$ obtained from the logic $\mathtt{ESO}$-$\mathtt{HORN}$ tailored by Grädel~\cite{Gradel07EATCS} to characterize $\mathtt{P}$, by restricting both the number of first-order variables and the arity of second-order predicate symbols. 
Besides, it includes an additional restriction – the ``monotonicity condition’’ – that reflects the geometrical consideration above-mentioned. \cite{BacqueyGO17ICALP} proves that this logic exactly characterizes the \emph{linear time} complexity class of cellular automata: more precisely, for each integer $d\ge 1$, a set $L$ of $d$-dimensional pictures can be decided in linear time on a $d$-dimensional CA – written $L \in \mathtt{DLIN}^d_{\mathtt{CA}}$ – if and only if it can be expressed in $\mathtt{monot}$-$\mathtt{ESO}$-$\mathtt{HORN}^d(\forall^{d+1}, \mathtt{arity}^{d+1})$. For short: 
\begin{center}
$\mathtt{DLIN}^d_{\mathtt{CA}} = \mathtt{monot}$-$\mathtt{ESO}$-$\mathtt{HORN}^d(\forall^{d+1}, \mathtt{arity}^{d+1})$. 
\end{center}
To summarize, expressing a concrete problem in this logic -- which seems an \emph{easy} task in practice and also is a \emph{necessary and sufficient} condition according to the above equality -- \emph{guarantees} that this problem can be solved in \emph{linear time} on a CA; moreover, the Horn formula that defines the problem can be \emph{automatically} translated into a program of CA that computes it in \emph{linear time}.

\subsection{
Logics for minimal time of cellular automata?}

\noindent
At this point, two natural questions arise:
\begin{enumerate}
\item Besides linear time, a robust and very expressive complexity class, what are the other \emph{significant} and \emph{robust} complexity classes of CA?
\item Can we exhibit characterizations of those complexity classes in some naturally (syntactically) defined logics so that any definition of a problem in such a logic can be \emph{automatically} translated into a program of the complexity in concern?
\end{enumerate}

Besides \emph{linear time}, the main complexity notion well-studied for a long time in the literature of CA is \emph{real-time}, i.e., \emph{minimal time}~\cite{Cole69rtIA,Smith72JCSS,Dyer80IC}. A cellular automaton is said to run in \emph{real-time} if it stops, i.e., gives the answer yes or no, at the \emph{minimal} instant when the elapsed time is sufficient for the output cell (the cell that gives the answer) to have received \emph{each} letter of the input. Real-time complexity appears as a \emph{sensitive/fragile} notion 
and one generally thinks it is so for CA of dimension 2 or more~\cite{Terrier99vNvsM},~\cite{Anael17DLT}. However, maybe surprisingly, one knows that real-time complexity is a \emph{robust} notion for \emph{one-dimensional} CA in the following sense: according to the many \emph{natural} variants of the definition of a one-dimensional CA, which essentially rest on the choice of the \emph{neighborhood} of the CA and the \emph{parallel} or \emph{sequential} presentation of its input word, \emph{exactly three} real-time classes of one-dimensional CA\footnote{By default, a CA has a \emph{two-way} communication and a \emph{parallel input} mode. Any CA (resp. \emph{one-way} CA or OCA) with \emph{sequential input}  mode is also called an \emph{iterative array} or $\mathtt{IA}$ (resp. OIA).} have been proved to be distinct~\cite{Cole69rtIA,ChoffCulik84AI,IbarraJiang87oneWayCA,Poupet05STACS}:
\begin{enumerate}
\item 
$\mathtt{RealTime}_{\mathtt{CA}} = \mathtt{RealTime}_{\mathtt{OIA}}$;
\item 
$\mathtt{Trellis} = \mathtt{RealTime}_{\mathtt{OCA}}$;
 \item 
$\mathtt{RealTime}_{\mathtt{IA}}$.
\end{enumerate}
The final and decisive step to establish this classification is a nice dichotomy of~\cite{Poupet05STACS} on \emph{admissible} neighborhoods\footnote{The \emph{neighborhood} of a CA is the finite set of integers $\mathcal{N}$ such that the state of any cell $x$ at any non-initial instant $t$ is determined by the states of the cells $x+d$, for $d\in \mathcal{N}$, at instant $t-1$. A neighborhood is \emph{admissible} with respect to a fixed output cell   (in general the first or the last cell) if it allows to communicate each bit of the input to the output cell.}
of CA, which can be rephrased as follows: for each neighborhood $\mathcal{N}$ admissible with respect to the first cell as output cell, the real-time complexity class of one-dimensional CA with parallel input mode and neighborhood~$\mathcal{N}$,
\begin{itemize}
\item either is equal to the real-time class for the neighborhood $\{-1,0,1\}$, 
i.e., $\mathtt{RealTime}_{\mathtt{CA}}$ (class 1 above),
\item or is equal to the real-time class for the neighborhood $\{-1,0\}$, i.e., $\mathtt{Trellis}$ (class 2 above).
\end{itemize}

\noindent
Further, it is \emph{surprising} to notice that 
 \begin{itemize}
\item 
the mutual relations between those three real-time classes are  \emph{wholly elucidated}: classes $\mathtt{Trellis}$ and $\mathtt{RealTime}_{\mathtt{IA}}$ are mutually \emph{incomparable for inclusion} whereas we have the strict inclusion
$\mathtt{Trellis} \cup \mathtt{RealTime}_{\mathtt{IA}} \subsetneqq \mathtt{RealTime}_{\mathtt{CA}}$~\cite{Cole69rtIA,CulikGS84SysTA,Terrier96uvu},

\item while it is \emph{unknown} whether the trivial inclusion 
$\mathtt{RealTime}_{\mathtt{CA}}\subseteq \mathtt{DLIN}^1_{\mathtt{CA}}$ is strict; 
worse, even whether the inclusion 
$\mathtt{RealTime}_{\mathtt{CA}}\subseteq \mathtt{LinSpace}$ is strict is an open problem!

\end{itemize}
\subsection{
Logics and grid circuits for real-time classes}\label{results}

\noindent
Each of the three real-time classes 1-3 is \emph{robust}, i.e, is not modified for many \emph{variants} of CA (change of neighborhoods, etc.) and has two or three \emph{quite different} equivalent definitions. For example, $\mathtt{RealTime}_{\mathtt{CA}}$ is equal to the \emph{linear time} class of one-way CA with parallel input mode. 
Similarly, \cite{Okhotin04rairo} has proved the surprising result that $\mathtt{Trellis}$ \emph{is} the class of languages generated by \emph{linear conjunctive grammars} (see also~\cite{Okhotin13CSR}) and \cite{Terrier03b} has established that a language $L$ is in $\mathtt{RealTime}_{\mathtt{IA}}$ if and only if its \emph{reverse} language $L^R$ is recognized in \emph{real time} by an \emph{alternating automaton with one counter}.

Logics have two nice and complementary properties: 
they are \emph{flexible}, hence \emph{expressive}; they have \emph{normal forms}, hence can be tailored for \emph{efficient programming}.
The main idea that led us to conceive the  \emph{different} logics for real-time classes can be summarized by the following simple question: what are the  \emph{different} ways to decide a language on a \emph{square grid circuit}? For any integer $n\ge 1$, let $C_n$ be the grid circuit $n\times n$ where the state $q\in Q$ (for finite~$Q$) of any site $(i,j)$, $1\le i,j \le n$, is determined by the states of its ``predecessors'' $(i-1,j)$, if it exists ($i>1$), and $(i,j-1)$, if it exists ($j>1$), so that the output cell is the site $(n,n)$. Up to symmetries, there are three canonical ways\footnote{To be complete, one should say that there is a fourth arrangement: place the input word on a \emph{side containing} the output cell. In this case, the grid circuit behaves like a \emph{finite automaton} (CA of dimension zero)...} to arrange an input word $w=w_1\ldots w_n$ of length $n$ on the grid $C_n$, see Figure~\ref{fig_grid}:
\begin{enumerate}
\item 
$\mathtt{GRID}_1$: place the input on any \emph{side} (or, equivalently, on both sides) that does (do) not contain the output cell; 
\item 
$\mathtt{GRID}_2$: place the input on the \emph{diagonal opposite} to the output cell; 
\item 
$\mathtt{GRID}_3$: place the input on the \emph{diagonal that contains} the output cell.   
\end{enumerate}


\begin{figure}
\centering
\begin{tikzpicture}[scale=0.5,every node/.style={scale=0.7}]

\foreach \x in {1,2,...,4}{
\foreach \y in {1,2,...,5}{
\draw[-to,shorten >=-1pt,thick](\x,\y) to (\x+0.7,\y);
\draw[-to,shorten >=-1pt,thick](\y,\x) to (\y,\x+0.7);
}}
\foreach \x in {1,2,...,5}{
\foreach \y in {1,2,...,5}{
\filldraw[fill=white] (\x,\y) circle (0.3cm) ;
}}
\draw (5,5) circle (0.35cm) ;
\foreach \x in {1,2,...,5}{
\draw(1,\x) node{$w_{\x}$};
\draw(\x,1) node{$w_{\x}$};
}
\draw(3,0) node{$\mathtt{GRID}_1$};


\foreach \x in {1,2,...,4}{
\foreach \y in {1,2,...,5}{
\draw[-to,shorten >=-1pt,thick](6+\x,\y) to (6+\x+0.7,\y);
\draw[-to,shorten >=-1pt,thick](6+\y,\x) to (6+\y,\x+0.7);
}}
\foreach \x in {1,2,...,5}{
\foreach \y in {1,2,...,5}{
\filldraw[fill=white] (6+\x,\y) circle (0.3cm) ;
}}
\draw (6+5,5) circle (0.35cm) ;
\foreach \x in {1,2,...,5}{
\draw(6+\x,6-\x) node{$w_{\x}$};
}
\draw(9,0) node{$\mathtt{GRID}_2$};


\foreach \x in {1,2,...,4}{
\foreach \y in {1,2,...,5}{
\draw[-to,shorten >=-1pt, thick](12+\x,\y) to (12+\x+0.7,\y);
\draw[-to,shorten >=-1pt, thick](12+\y,\x) to (12+\y,\x+0.7);
}}
\foreach \x in {1,2,...,5}{
\foreach \y in {1,2,...,5}{
\filldraw[fill=white] (12+\x,\y) circle (0.3cm) ;
}}
\draw (12+5,5) circle (0.35cm) ;
\foreach \x in {1,2,...,5}{
\draw(12+\x,\x) node{$w_{\x}$};
}
\draw(15,0) node{$\mathtt{GRID}_3$};
\end{tikzpicture}
\caption{\label{fig_grid}The three ways to arrange the input on the grid}
\end{figure}
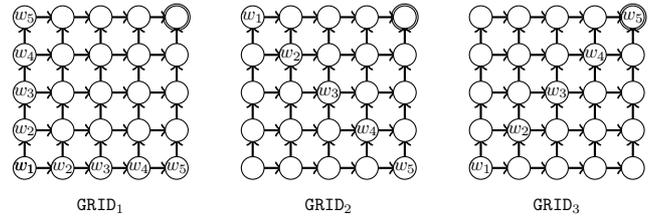


A simple (reversible) deformation transforms a grid circuit of $\mathtt{GRID}_i$, $i=1,2,3$, into a time-space diagram of a CA of the real-time class $i$ in concern (recall: 1:~$\mathtt{RealTime}_{\mathtt{CA}}$; 2: $\mathtt{Trellis}$; 3: $\mathtt{RealTime}_{\mathtt{IA}}$), and conversely.
More precisely, to characterize the three real-time classes, we define three sub-logics of the Horn logic that characterizes linear time of one-dimensional CA ($\mathtt{DLIN}^1_{\mathtt{CA}} = \mathtt{monot}$-$\mathtt{ESO}$-$\mathtt{HORN}^1(\forall^{2}, \mathtt{arity}^{2})$), called respectively $\mathtt{pred}$-$\mathtt{ESO}$-$\mathtt{HORN}$, $\mathtt{incl}$-$\mathtt{ESO}$-$\mathtt{HORN}$
and $\mathtt{pred}$-$\mathtt{dio}$-$\mathtt{ESO}$-$\mathtt{HORN}$ (defined in the next section), and we prove the following equalities:
\begin{enumerate}
\item $\mathtt{pred}$-$\mathtt{ESO}$-$\mathtt{HORN}=\mathtt{GRID}_1=\mathtt{RealTime}_{\mathtt{CA}}$
\item $\mathtt{incl}$-$\mathtt{ESO}$-$\mathtt{HORN}=\mathtt{GRID}_2=\mathtt{Trellis}$
\item $\mathtt{pred}$-$\mathtt{dio}$-$\mathtt{ESO}$-$\mathtt{HORN}=\mathtt{GRID}_3=\mathtt{RealTime}_{\mathtt{IA}}$
\end{enumerate}

To establish the double nature of our three logics and deduce the previous equalities 1-3, we present each logic in two forms:
\begin{itemize}
\item we define it the \emph{largest possible}, showing the extent of its \emph{expressiveness};
\item we prove for it the \emph{most restricted normal form}.
\end{itemize} 
In each case, a formula in normal form can be \emph{translated literally} into a grid program, which is essentially a CA of the \emph{real-time} complexity class in concern.

\emph{Structure of the paper:} In preliminaries, we recall the classical definitions of one-dimensional cellular automata and of their real-time classes and define our three logics with an example of a problem naturally expressed  in such a logic. Section~\ref{sec:normal} establishes how each of our logics can be normalized. Using these normal forms we show in Section~\ref{sec:Logic_CA} that our three logics exactly characterize the three real-time complexity classes and also -- for inclusive logic -- the class of linear conjunctive languages of Okhotin~\cite{Okhotin04rairo}. Section~\ref{sec:conc} gives conclusive remarks.

\section{Preliminaries}

\subsection{Cellular automata and real-time complexity}

A cellular automaton in one dimension is a tape of cells (each cell is a finite automaton) indexed by $\mathbb{Z}$. Each cell takes a value from a finite set of states and the cells evolve synchronously along a discrete time scale. The evolution of the cell $c$ is done according to a transition function which takes as input the state of the cell itself and the states of its neighbors at the previous time step and outputs a new state for the cell. 
\smallskip
\begin{definition}[cellular automaton]\label{Def_CA}
A \emph{cellular automaton} is defined by a 5-tuple $(Q,\Sigma,Q_{accept},\mathcal{N},\delta)$ where $Q$ is a finite \emph{set of states}, $\Sigma \subset Q$ is the input alphabet, $Q_{accept} \subset Q$ is the set of \emph{accepting states}, the \emph{neighborhood} $\mathcal{N}$ is a finite ordered subset of $\mathbb{Z}$ and $\delta$ is the \emph{transition function} from $Q^{|\mathcal{N}|}$ to $Q$. The state of the cell $c$ at time $t$ is denoted by $\langle c,t \rangle$. The state of the cell $c$ at time $t>1$ is defined by the transition function: $ \langle c,t \rangle = \delta (\langle c+v,t-1 \rangle: v \in \mathcal{N})$.
\end{definition}

\smallskip
\begin{definition}[real time language acceptance]\label{Def_RealTime}
  Cellular automata can act as \emph{language acceptors}. In this case cellular automata work on a finite set of cells indexed by $[1,n]$ where $n$ is the size of the input word $w=w_1\dots w_n$, one of this cells is also chosen as the \emph{output cell} (usually one of the border cells). A word $w$ is said to be \emph{accepted} by an automaton in real-time if its output cell enters an accepting state immediately after getting all the information from $w$. The \emph{language accepted} by a cellular automaton $\mathcal{A}$ in real-time (denoted by $L(\mathcal{A})$) is the set of all its accepted words in real-time. 
\end{definition}

\smallskip
\noindent\emph{\bf Convention }(permanent state, quiescent state).\\
All the cells of index outside $[1,n]$ are in the \emph{permanent state} $\sharp$. Without information from the input a cell of index $[1,n]$ is in the \emph{quiescent state} $\lambda$. \\

The real-time computation power of a CA only depends on its \emph{communication scheme}. That is fully determined by the following three specifications: the way the input is fed to the automaton, the way the cells communicate (depending on the neighborhood) and the output cell. The input is usually fed to the automaton in a parallel way: the $i^{th}$ bit of the input is given to the $i^{th}$ cell at the start of the computation. The input can also be fed in a sequential way: the $i^{th}$ bit of the input is given to the first cell at time $i$ for which we add a specific transition function $\delta_{input}$. Usually the output cell is the first cell for two-way communications and the last one for one-way communication.

\smallskip
CA have their input fed in a parallel way and the cells communicate in two-way mode ($\mathcal{N}=\{-1,0,1\}$). One-way cellular automata (OCA) and iterative arrays (IA) are two natural variants of CA. OCA are narrowed on the way the cells communicate: the information is only transmitted from left to right ($\mathcal{N}=\{-1,0\}$). The input mode of IA is no more parallel but sequential.

\begin{definition}[$\mathtt{RealTime_{CA}}$]\label{Def_RealTimeCA}
The class $\mathtt{RealTime_{CA}}$ is the set of languages accepted by real-time CA with a parallel input, the neighborhood $\mathcal{N}=\{-1,0,1\}$ and the first cell as the output cell. 
\end{definition}
The class $\mathtt{RealTime_{CA}}$ is equivalent to $\mathtt{RealTime_{OIA}}$, the set of languages accepted by one-way IA with sequential input running in real-time with neighborhood $\mathcal{N}=\{-1,0\}$ and the last cell as the output cell.

\smallskip

\begin{definition}[$\mathtt{Trellis}$]\label{Def_RealTimeTrellis}
The class $\mathtt{Trellis}$ is the set of languages accepted by trellis automata or equivalently by OCA running in real-time with a parallel input, the neighborhood $\mathcal{N}=\{-1,0\}$ and the last cell as the output cell.
\end{definition}

\smallskip
\begin{definition}[$\mathtt{RealTime_{IA}}$]\label{Def_RealTimeIA}
The class $\mathtt{RealTime_{IA}}$ is the set of languages accepted by IA running in real-time with a sequential input, the neighborhood $\mathcal{N}=\{-1,0,1\}$ and the first cell as the output cell.
\end{definition}

The space-time diagrams of these real-time classes are depicted in Figure~\ref{fig_CA}.


\begin{figure}
\centering
\begin{tikzpicture}[scale=0.4,every node/.style={scale=0.7}]


\foreach \x in {1,2,...,4}{
\foreach \y in {\x,...,4}{
\draw[-to,shorten >=-1pt,thick](2+\x,7+5-\y) to (2+\x,7+5-\y+0.7);
}}

\foreach \x in {1,2,...,3}{
\foreach \y in {\x,...,3}{
\draw[-to,shorten >=-1pt,thick](2+\x,7+4-\y) to (2+\x+0.8,7+4-\y+0.8);
}}

\foreach \x in {1,2,...,4}{
\foreach \y in {\x,...,4}{
\draw[-to,shorten >=-1pt,thick](2+\x+1,7+5-\y) to (2+\x+0.2,7+5-\y+0.8);
}}

\foreach \x in {1,2,...,5}{
\foreach \y in {\x,...,5}{
\filldraw[fill=white] (2+\x,7+6-\y) circle (0.3cm) ;
}}

\filldraw[fill=gray](2+1,7+5) circle (0.3cm) ;
\foreach \x in {1,2,...,5}{
\draw[->](2+\x,7+0.4) to (2+\x,7+0.7);
\draw(2+\x,7+0) node{$w_{\x}$};
}
\draw(2+3,7-1) node{$\mathtt{RealTime_{CA}}$};

\draw[implies-implies,double equal sign distance] (2+4,10) -- (2+7,10);


\foreach \x in {1,2,...,5}{
  \foreach \y in {1,2,...,4}{
    \draw[-to,shorten >=-1pt,thick](3+7+\x,5+\x+\y) to (3+7+\x,5+\x+\y+0.7);
  }}
\foreach \x in {1,2,...,4}{
  \foreach \y in {1,2,...,5}{
    \draw[-to,shorten >=-1pt,thick](3+7+\x,5+\x+\y) to (3+7+\x+0.75,5+\x+\y+0.75);
  }}

\foreach \x in {1,2,...,5}{
  \foreach \y in {1,2,...,5}{
    \filldraw[fill=white] (3+7+\x,5+\x+\y) circle (0.3cm) ;
}}

\filldraw[fill=gray] (3+7+5,5+10) circle (0.3cm) ;
\foreach \x in {1,2,...,5}{
\draw[->](3+7+0.4,6+\x) to (3+7.7,6+\x);
\draw(3+7,6+\x) node{$w_{\x}$};
}
\draw(3+10,7-1) node{$\mathtt{RealTime_{OIA}}$};


\foreach \y in {1,2,...,4}{
\foreach \x in {\y,...,4}{
  \draw[-to,shorten >=-1pt,thick](\x,\y) to (\x+0.75,\y+0.75);
}}


\foreach \y in {1,2,...,4}{
\foreach \x in {\y,...,4}{
\draw[-to,shorten >=-1pt,thick](\x+1,\y) to (\x+1,\y+0.7);
}}

\foreach \y in {1,2,...,5}{
\foreach \x in {\y,...,5}{
\filldraw[fill=white] (\x,\y) circle (0.3cm) ;
}}

\filldraw[fill=gray](5,5) circle (0.3cm) ;
\foreach \x in {1,2,...,5}{
\draw[->](\x,0.4) to (\x,0.7);
\draw(\x,0) node{$w_{\x}$};
}
\draw(3,-1) node{$\mathtt{RealTime_{OCA}}$};

\draw[implies-implies,double equal sign distance] (6,3) -- (8,3);

\draw(0,-1.5) rectangle (13,5.5);
\draw(0,5.5) rectangle (18,16);
\draw(13,-1.5) rectangle (18,5.5);

\foreach \y in {1,2,...,4}{
\foreach \x in {\y,...,4}{
\draw[-to,shorten >=-1pt,thick](7+0.5-0.5*\y+\x,\y) to (8+\x-0.5*\y-0.2,\y+0.75);
}}

\foreach \y in {1,2,...,4}{
\foreach \x in {\y,...,4}{
\draw[-to,shorten >=-1pt,thick](7+\x+1.5-0.5*\y,\y) to (8-0.5*\y+\x+0.2,\y+0.75);
}}

\foreach \y in {1,2,...,5}{
\foreach \x in {\y,...,5}{
\filldraw[fill=white] (7.5+\x-0.5*\y,\y) circle (0.3cm) ;
}}

\filldraw[fill=gray](7+3,5) circle (0.3cm) ;
\foreach \x in {1,2,...,5}{
\draw[->](7+\x,0.4) to (7+\x,0.7);
\draw(7+\x,0) node{$w_{\x}$};
}
\draw(7+3,-1) node{$\mathtt{Trellis}$};


\foreach \y in {1,2,...,4}{
\draw[-to,shorten >=-1pt, thick](14+1,-1+\y) to (14+1,-1+\y+0.7);
}

\foreach \y in {1,2}{
\draw[-to,shorten >=-1pt, thick](14+2,-1+\y+1) to (14+2,-1+\y+1+0.7);
}

\foreach \y in {1,2}{
\draw[-to,shorten >=-1pt, thick](14+\y,-1+\y) to (14+\y+0.75,-1+\y+0.75);
\draw[-to,shorten >=-1pt, thick](14+4-\y,2+-1+\y) to (14+3.25-\y,-1+\y+2+0.75);
}

\foreach \y in {1,2}{
\draw[-to,shorten >=-1pt, thick](14+1,-1+\y+1) to (14+1+0.75,-1+\y+1+0.75);
\draw[-to,shorten >=-1pt, thick](14+2,-1+\y+1) to (14+1.25,-1+\y+1+0.75);
}

\foreach \y in {1,2,...,5}{
\filldraw[fill=white] (14+1,-1+\y) circle (0.3cm) ;
}
\foreach \y in {1,2,...,3}{
\filldraw[fill=white] (14+2,-1+\y+1) circle (0.3cm) ;
}
\filldraw[fill=white] (14+3,-1+3) circle (0.3cm) ;

\filldraw[fill=gray](14+1,-1+5) circle (0.3cm) ;
\foreach \x in {1,2,...,5}{
\draw[->](14+0.4,-1+\x) to (14.7,-1+\x);
\draw(14,-1+\x) node{$w_{\x}$};
}
\draw(15,-1) node{$\mathtt{RealTime_{IA}}$};
\end{tikzpicture}
\caption{\label{fig_CA}The space-time diagram of the three natural real-time classes}
\end{figure}
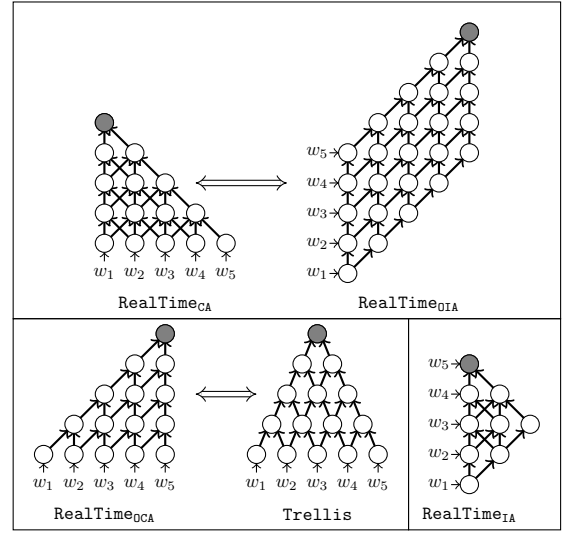

\subsection{Our logics}

\noindent
The ``local'' nature of our logics requires that the underlying structure encoding an input word $w=w_1\ldots w_n$ on its index interval $[1,n]=\{1,\ldots,n\}$ only uses the \emph{successor} and \emph{predecessor} functions and the monadic predicates $\mathtt{min}$ and $\mathtt{max}$ as its  \emph{only} arithmetic functions/predicates:
\begin{definition}[structure encoding a word]
Each nonempty word $w=w_1\ldots w_n\in\Sigma^n$ on a fixed finite alphabet $\Sigma$ is represented by the first-order structure 
\[
\langle w \rangle :=
([1,n] ; (Q_s)_{s\in\Sigma}, \mathtt{min},\mathtt{max}, \mathtt{suc},
\mathtt{pred})
\]
of domain $[1,n]$, monadic predicates $Q_s$, $s\in\Sigma$, $\mathtt{min}$ and $\mathtt{max}$
such that
$Q_s(i)\iff w_i=s$,  
$\mathtt{min}(i) \iff i=1$, and $\mathtt{max}(i) \iff i=n$, 
and unary functions $\mathtt{suc}$ and $\mathtt{pred}$ such that
$\mathtt{suc}(i)=i+1$ for $i<n$ and $\mathtt{suc}(n)=n$,
 $\mathtt{pred}(i)=i-1$ for $i>1$ and $\mathtt{pred}(1)=1$. 
 Let $\mathcal{S}_{\Sigma}$ denote the signature 
 $\{(Q_s)_{s\in\Sigma}, \mathtt{min},\mathtt{max}, \mathtt{suc}, \mathtt{pred}\}$ 
of structure $\langle w \rangle$. 
The monadic predicates $Q_s$, $s\in\Sigma$, $\mathtt{min}$, and $\mathtt{max}$ 
of $\mathcal{S}_{\Sigma}$ 
are called \emph{input predicates}.
\end{definition}
Let $x+k$ and $x-k$ abbreviate the terms $\mathtt{suc}^k(x)$ and $\mathtt{pred}^k(x)$, for a fixed integer $k\ge 0$.

\medskip
Let us now define two of our logics:

\begin{definition}[predecessor logics] \label{def:predlogics} A \emph{predecessor Horn formula} (resp. \emph{predecessor Horn formula with diagonal input-output}) is a formula of the form $\Phi = \exists\mathbf{R}\forall x \forall y \psi(x,y)$ where $\psi$ is a conjunction of Horn clauses on the variables $x,y$,\\
-- of signature $\mathcal{S}_{\Sigma}\cup \mathbf{R}$ (resp. $\mathcal{S}_{\Sigma}\cup \mathbf{R} \cup \{=\}$) where $\mathbf{R}$ is a set of binary predicates called \emph{computation predicates},\\
-- of the form $\delta_1\land\ldots\land\delta_r\to\delta_0$ where the conclusion $\delta_0$ is either a \emph{computation atom} $R(x,y)$ with $R\in\mathbf{R}$, or $\bot$ (\emph{False}) and each hypothesis $\delta_i$ is either an \emph{input literal/conjunction} 
of one of the forms: 
\begin{itemize}
\item $Q_s(x-a)$, $Q_s(y-a)$ (resp. $Q_s(x-a)\land x=y$), for $s\in\Sigma$ and an integer $a\ge 0$,
\item  $U(x-a)$, $\lnot U(x-a)$, $U(y-a)$ or $\lnot U(y-a)$, for $U\in \{\mathtt{min},\mathtt{max}\}$ and an integer $a\ge 0$, 
\end{itemize}

or a \emph{computation atom} of the form $S(x-a,y-b)$ or $S(y-b,x-a)$, for $S\in\mathbf{R}$ and some integers $a,b\ge 0$. 
 \end{definition}
Let $\mathtt{pred}$-$\mathtt{ESO}$-$\mathtt{HORN}$ 
(resp. $\mathtt{pred}$-$\mathtt{dio}$-$\mathtt{ESO}$-$\mathtt{HORN}$)
denote the class of \emph{predecessor Horn formulas} (resp. \emph{predecessor Horn formulas with diagonal input-output}) and, by abuse of notation, the class of languages they define. 

\smallskip
The formulas of the ``predecessor'' logics defined above use the \emph{predecessor} function but \emph{not} the \emph{successor} function: both logics inductively define problems in \emph{increasing} both coordinates $x$ and $y$. The inductive principle of our last logic is seemingly different: it lies on \emph{inclusions} of intervals $[x,y]$.

\smallskip
\begin{definition}[inclusion logic]\label{def:incllogic}
An \emph{inclusion Horn formula} is a formula of the form $\Phi = \exists\mathbf{R}\forall x \forall y \psi(x,y)$ where $\psi$ is a conjunction of Horn clauses 
of signature $\mathcal{S}_{\Sigma}\cup \mathbf{R}\cup \{=,\le,<\} $ where $\mathbf{R}$ is a set of binary predicates called \emph{computation predicates},
of the form $x\le y\land \delta_1\land\ldots\land\delta_r\to\delta_0$ where the conclusion $\delta_0$ is 
either a \emph{computation atom} $R(x,y)$ with $R\in\mathbf{R}$, or the atom $\bot$ (\emph{False}), and each hypothesis $\delta_i$~is 
\begin{enumerate}
\item either an \emph{input literal} of the form~\footnote{Without loss of generality, assume that there is no negation over a predicate $Q_s$.} 
$U(x+a)$, $\lnot U(x+a)$, $U(y+a)$ or $\lnot U(y+a)$, for $U\in \{(Q_s)_{s\in\Sigma}, \mathtt{min},\mathtt{max}\}$ and an integer $a\in\mathbb{Z}$,
\item or the (in)equality $x=y$ or $x<y$~\footnote{
Then, the hypothesis $x\le y$ is redundant.},
\item\label{interval} or a conjunction of the form
\begin{center}
 $
 S(x+a,y-b)\land x+a\le y-b
$
\end{center}
for a \emph{computation atom} $S(x+a,y-b)$, with $S\in\mathbf{R}$ and some integers $a,b\ge 0$.
 \end{enumerate}
 \end{definition}

\noindent
Let $\mathtt{incl}$-$\mathtt{ESO}$-$\mathtt{HORN}$ 
denote the class of \emph{inclusion Horn formulas} and, also, the class of languages they define. 

\medskip
Note that the ``inclusion'' meaning of logic $\mathtt{incl}$-$\mathtt{ESO}$-$\mathtt{HORN}$ is given by hypotheses $x\le y$ and~$x+a\le y-b$.
It means that the inductive computation of each value $R(x,y)$, for $x\le y$ and $R\in\mathbf{R}$, only use values of the form \linebreak$S(x+a,y-b)$, for $S\in\mathbf{R}$ and an \emph{included} interval $[x+a,y-b]\subseteq [x,y]$.

\smallskip
\noindent
\emph{Notation:} We will freely use the intuitive abbreviations $x~>~a$, $x=a$, for a constant integer $a\ge 1$, and $x \le n-a$, $x<n-a$, $x=n-a$, for a constant integer $a\ge 0$, and similarly for $y$. 
For example, $x>3$ is written in place of $\lnot \mathtt{min}(x-2)$
and $y\le n-2$ is written in place of $\lnot \mathtt{max}(y+1)$. 

\smallskip
\noindent
\emph{Technical remarks about our logics:} Without loss of generality, we can suppose that each clause having a hypothesis atom of the form $S(x-a,y-b)$ or $S(y-b,x-a)$, for $a,b\ge 0$, has \emph{also} the hypotheses 
$x>a$ (if $a>0$) and $y>b$ (if $b>0$). The same for each hypothesis atom of the form $Q_s(x-a)$ or $Q_s(y-b)$, for $a,b>0$. Similarly, we assume that each clause with a hypothesis of the form $Q_s(x+a)$ (resp. $Q_s(y+a)$), with $a>0$, \emph{also} contains the hypothesis $x\le n-a$ (resp. $y\le n-a$). Similarly, for each atom $S(x+a,y-b)$, for $a,b\ge 0$.

\smallskip
\noindent
\emph{Comparing the input presentation in our logics:} 
The presentation of the input is \emph{more restrictive} in Definition~\ref{def:predlogics} of predecessor logics than in that of inclusion logic (Definition~\ref{def:incllogic}) because we have forbidden the use of the successor function for uniformity/aesthetics.
However, allowing the same \emph{largest} set of input literals $(\lnot)U(x+a)$, $(\lnot)U(y+a)$, 
for $U\in \{(Q_s)_{s\in\Sigma}, \mathtt{min},\mathtt{max}\}$ and $a\in\mathbb{Z}$, 
\emph{does not modify} the expressive power of predecessor logics: steps~5 and~6 of the normalization of inclusive logic in Section~\ref{sec:normal} can be easily adapted to predecessor logics.

\subsection*{Using our logics for programming: an example}

\noindent
We now give an example  of a natural problem expressed in one of our logics. 
\smallskip
The language $\mathtt{notBordered}$ is the set of all words $w\in\Sigma^+$ with no proper prefix equal to a proper suffix:\\
$\mathtt{notBordered}=\{w \mid \forall w'$ such that $w=w'u=vw'$,  $w'= \epsilon$ or $w'=w\}$.\\
This language can be defined by $\Phi_{\mathtt{notBordered}} \coloneqq \exists \mathtt{Border}\forall x \forall y \psi$ of $\mathtt{pred}$-$\mathtt{ESO}$-$\mathtt{HORN}$ where $\psi$ is the conjunction of the following clauses where $\mathtt{Border}(x,y)$ means $w_1\dots w_x=w_{y-x+1}\dots w_y$ (see Figure~\ref{fig_notBordered}):
\begin{enumerate}
\item $\mathtt{min}(x) \land \lnot\mathtt{min}(y) \land Q_s(x) \land Q_s(y) \to \mathtt{Border}(x,y)$
\item $\lnot\mathtt{min}(x) \land \lnot\mathtt{min}(y) \land \mathtt{Border}(x-1,y-1) \land Q_s(x) \land Q_s(y) \to \mathtt{Border}(x,y)$, for all $s\in\Sigma$;
\item $\mathtt{max}(y) \land \mathtt{Border}(x,y) \to \bot$.
\end{enumerate}
If $\mathtt{Border}(x,y)$ is true when $y$ is maximal, then $w_1\dots w_x=w_{n-x+1}\dots w_n$ and therefore $w \notin \mathtt{notBordered}$.

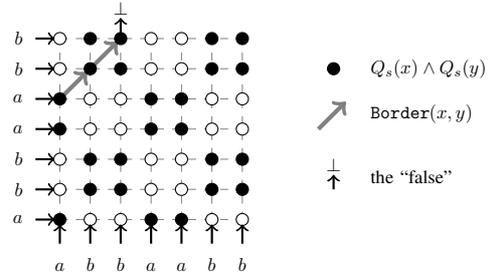
\begin{figure}[h]
\centering
\begin{tikzpicture}[scale=0.4,every node/.style={scale=0.7}]
\draw [gray,dashed, very thin] (0,0) grid (6,6);
\draw[-to,shorten >=-1pt,thick](-0.8,0) to (-0.25,0);
\draw[left](-1,0) node{$ a$};
\draw[-to,shorten >=-1pt,thick](-0.8,1) to (-0.25,1);
\draw[left](-1,1) node{$ b$};
\draw[-to,shorten >=-1pt,thick](-0.8,2) to (-0.25,2);
\draw[left](-1,2) node{$ b$};
\draw[-to,shorten >=-1pt,thick](-0.8,3) to (-0.25,3);
\draw[left](-1,3) node{$ a$};
\draw[-to,shorten >=-1pt,thick](-0.8,4) to (-0.25,4);
\draw[left](-1,4) node{$ a$};
\draw[-to,shorten >=-1pt,thick](-0.8,5) to (-0.25,5);
\draw[left](-1,5) node{$ b$};
\draw[-to,shorten >=-1pt,thick](-0.8,6) to (-0.25,6);
\draw[left](-1,6) node{$ b$};
\draw[-to,shorten >=-1pt,thick](0,-0.8) to (0,-0.25);
\draw[below](0,-1) node{$\strut a$};
 \filldraw[fill=black] (0,0) circle (0.2cm) ;
 \filldraw[fill=white] (1,0) circle (0.2cm) ;
 \filldraw[fill=white] (2,0) circle (0.2cm) ;
 \filldraw[fill=black] (3,0) circle (0.2cm) ;
 \filldraw[fill=black] (4,0) circle (0.2cm) ;
 \filldraw[fill=white] (5,0) circle (0.2cm) ;
 \filldraw[fill=white] (6,0) circle (0.2cm) ;
\draw[-to,shorten >=-1pt,thick](1,-0.8) to (1,-0.25);
\draw[below](1,-1) node{$\strut b$};
 \filldraw[fill=white] (0,1) circle (0.2cm) ;
 \filldraw[fill=black] (1,1) circle (0.2cm) ;
 \filldraw[fill=black] (2,1) circle (0.2cm) ;
 \filldraw[fill=white] (3,1) circle (0.2cm) ;
 \filldraw[fill=white] (4,1) circle (0.2cm) ;
 \filldraw[fill=black] (5,1) circle (0.2cm) ;
 \filldraw[fill=black] (6,1) circle (0.2cm) ;
\draw[-to,shorten >=-1pt,thick](2,-0.8) to (2,-0.25);
\draw[below](2,-1) node{$\strut b$};
 \filldraw[fill=white] (0,2) circle (0.2cm) ;
 \filldraw[fill=black] (1,2) circle (0.2cm) ;
 \filldraw[fill=black] (2,2) circle (0.2cm) ;
 \filldraw[fill=white] (3,2) circle (0.2cm) ;
 \filldraw[fill=white] (4,2) circle (0.2cm) ;
 \filldraw[fill=black] (5,2) circle (0.2cm) ;
 \filldraw[fill=black] (6,2) circle (0.2cm) ;
\draw[-to,shorten >=-1pt,thick](3,-0.8) to (3,-0.25);
\draw[below](3,-1) node{$\strut a$};
  \filldraw[fill=black] (0,3) circle (0.2cm) ;
 \filldraw[fill=white] (1,3) circle (0.2cm) ;
 \filldraw[fill=white] (2,3) circle (0.2cm) ;
 \filldraw[fill=black] (3,3) circle (0.2cm) ;
 \filldraw[fill=black] (4,3) circle (0.2cm) ;
 \filldraw[fill=white] (5,3) circle (0.2cm) ;
 \filldraw[fill=white] (6,3) circle (0.2cm) ;
\draw[-to,shorten >=-1pt,thick](4,-0.8) to (4,-0.25);
\draw[below](4,-1) node{$\strut a$};
\draw[-to,shorten >=-1pt,gray,ultra thick] (0,4) -- (0.84,4.84);
  \filldraw[fill=black] (0,4) circle (0.2cm) ;
 \filldraw[fill=white] (1,4) circle (0.2cm) ;
 \filldraw[fill=white] (2,4) circle (0.2cm) ;
 \filldraw[fill=black] (3,4) circle (0.2cm) ;
 \filldraw[fill=black] (4,4) circle (0.2cm) ;
 \filldraw[fill=white] (5,4) circle (0.2cm) ;
 \filldraw[fill=white] (6,4) circle (0.2cm) ;
\draw[-to,shorten >=-1pt,thick](5,-0.8) to (5,-0.25);
\draw[below](5,-1) node{$\strut b$};
 \filldraw[fill=white] (0,5) circle (0.2cm) ;
\draw[-to,shorten >=-1pt,gray,ultra thick] (1,5) -- (1.84,5.84);
  \filldraw[fill=black] (1,5) circle (0.2cm) ;
 \filldraw[fill=black] (2,5) circle (0.2cm) ;
 \filldraw[fill=white] (3,5) circle (0.2cm) ;
 \filldraw[fill=white] (4,5) circle (0.2cm) ;
 \filldraw[fill=black] (5,5) circle (0.2cm) ;
 \filldraw[fill=black] (6,5) circle (0.2cm) ;
\draw[-to,shorten >=-1pt,thick](6,-0.8) to (6,-0.25);
\draw[below](6,-1) node{$\strut b$};
 \filldraw[fill=white] (0,6) circle (0.2cm) ;
 \filldraw[fill=black] (1,6) circle (0.2cm) ;
\draw[-to,shorten >=-1pt,thick](2,6.2) to (2,6.6);
\draw(2,7) node{$ \bot$};
 \filldraw[fill=black] (2,6) circle (0.2cm) ;
 \filldraw[fill=white] (3,6) circle (0.2cm) ;
 \filldraw[fill=white] (4,6) circle (0.2cm) ;
 \filldraw[fill=black] (5,6) circle (0.2cm) ;
 \filldraw[fill=black] (6,6) circle (0.2cm) ;

 \filldraw[fill=black] (9,5) circle (0.2cm) ;
 \draw[right](10,5) node{$Q_s(x)\land Q_s(y)$};

 \draw[-to,shorten >=-1pt,gray,ultra thick] (8.5,3) -- (9.34,3.84);
 \draw[right](10,3.5) node{$\mathtt{Border}(x,y)$};

 \draw[-to,shorten >=-1pt,thick](9,1) to (9,1.4);
 \draw(9,1.8) node{$ \bot$};
  \draw[right](10,1.4) node{the ``false''};

\end{tikzpicture}
\caption{\label{fig_notBordered}Computation of $\Phi_{\mathtt{notBordered}}$ on the word $abbaabb$}
\end{figure}

So, as a consequence of this paper, 
$\mathtt{notBordered}$ belongs to $\mathtt{RealTime}_{\mathtt{CA}}$.
In fact, more is known \cite{Terrier96uvu}:
\[
\mathtt{notBordered}\in
\mathtt{RealTime}_{\mathtt{CA}}\setminus (\mathtt{TREILLIS}\cup \mathtt{RealTime}_{\mathtt{IA}}).
\]

\section{Normalizing our logics}\label{sec:normal}

\noindent
The most difficult and main parts of the proofs of our descriptive complexity results, i.e., equalities 1-3 of Subsection~\ref{results}, are the following \emph{normalization lemmas}.
They are key ingredients because our normalized formulas can be \emph{literally simulated} by grids and finally by CA of the corresponding real-time classes.

\smallskip
\begin{lemma}[normalization of predecessor logics]\label{lemmapred}
Each formula $\Phi\in\mathtt{pred}$-$\mathtt{ESO}$-$\mathtt{HORN}$ 
(resp. $\Phi\in\mathtt{pred}$-$\mathtt{dio}$-$\mathtt{ESO}$-$\mathtt{HORN}$)
is equivalent to a formula  $\Phi'\in\mathtt{pred}$-$\mathtt{ESO}$-$\mathtt{HORN}$ 
(resp. $\Phi'\in\mathtt{pred}$-$\mathtt{dio}$-$\mathtt{ESO}$-$\mathtt{HORN}$)
\emph{normalized} as follows: each clause of $\Phi'$ is of one of the following forms:

\begin{itemize}

\item \emph{input clause} of the form, for $s\in\Sigma$  and $R\in\mathbf{R}$:\\
 $\mathtt{min}(x) \land \mathtt{min}(y)\land Q_s(y)\to R(x,y)$, or\\
 $\mathtt{min}(x) \land \lnot\mathtt{min}(y)\land Q_s(y)\to R(x,y)$\\
 (resp. $x=y\land \mathtt{min}(x)\land Q_s(x) \to R(x,y)$, or\\
 $x=y\land \lnot\mathtt{min}(x)\land Q_s(x) \to R(x,y)$).
 
\item the \emph{contradiction clause}, for a fixed $R_{\bot}\in\mathbf{R}$:
$\mathtt{max}(x)\land \mathtt{max}(y)\land R_{\bot}(x,y)\to \bot$;

\item  \emph{computation clause} of the form $\delta_1\land\ldots\land\delta_r\to R(x,y)$, for $R\in\mathbf{R}$, where each hypothesis $\delta_i$ 
is a conjunction of the form
$S(x-1,y) \land \lnot \mathtt{min}(x)$ or $S(x,y-1) \land \lnot \mathtt{min}(y)$, for $S\in\mathbf{R}$.

\end{itemize}

\end{lemma}
Let $\mathtt{normal}$-$\mathtt{pred}$-$\mathtt{ESO}$-$\mathtt{HORN}$
(resp. $\mathtt{normal}$-$\mathtt{pred}$-$\mathtt{dio}$-$\mathtt{ESO}$-$\mathtt{HORN}$)
denote the class of formulas (languages) so defined.

\smallskip
\begin{lemma}[normalization of inclusion logic]\label{lemmaincl}
Each formula $\Phi\in\mathtt{incl}$-$\mathtt{ESO}$-$\mathtt{HORN}$ is equivalent to a formula 
$\Phi'\in\mathtt{incl}$-$\mathtt{ESO}$-$\mathtt{HORN}$ \emph{normalized} as follows: each clause of $\Phi'$ is of one of the following forms:

\begin{itemize}

\item \emph{input clause} of the form $x=y \land Q_s(x)\to R(x,y)$,
for $s\in\Sigma$ and $R\in\mathbf{R}$;
 
\item the \emph{contradiction clause}, for a fixed $R_{\bot}\in\mathbf{R}$, 
$\mathtt{min}(x)\land \mathtt{max}(y)\land R_{\bot}(x,y)\to \bot$;

\item  \emph{computation clause} of the form\footnote{Note that the hypothesis $x<y$ is equivalent to the expected inequality $x+1\le y$ or $x \le y-1$.}
$x<y \land \delta_1\land\ldots\land\delta_r\to R(x,y)$,
where $R\in\mathbf{R}$ and where each hypothesis $\delta_i$ 
is a computation atom of either form 
$S(x+1,y)$ or $S(x, y-1)$, for $S \in\mathbf{R}$.

\end{itemize}

\end{lemma}
Let $\mathtt{normal}$-$\mathtt{incl}$-$\mathtt{ESO}$-$\mathtt{HORN}$ denote the class of formulas (resp. languages) so defined.

\subsection*{Proof of the normalization lemmas~\ref{lemmapred} and~\ref{lemmaincl}}
\noindent
The normalization processes of our three logics are quite similar each other; further, some steps are exactly the same. Therefore, we choose to present here below the successive normalization steps for one logic:
$\mathtt{pred}$-$\mathtt{ESO}$-$\mathtt{HORN}$. 
Afterwards, we will succinctly describe how those steps should be adapted for the two other logics.

\subsubsection*{Proof of Lemma~\ref{lemmapred}: Normalization of predecessor Horn formulas}

Let a formula $\Phi\in\mathtt{pred}$-$\mathtt{ESO}$-$\mathtt{HORN}$.
For simplicity of notation, we first assume that the only computation atoms of $\Phi$ are of the form 
$R(x-a,y-b)$, $a,b\ge 0$ (no atom of the form $R(y-b,x-a)$). We will show at the end of the proof how to manage the general case. $\Phi$ will be transformed into an equivalent normalized form $\Phi'$ by a sequence of 10  steps:
\begin{enumerate}
\item Processing the contradiction clauses;
\item Processing the input;
\item Restriction of computation atoms to $R(x-1,y)$, $R(x,y-1)$, and $R(x,y)$;
\item Elimination of atoms $x>a$, $x=a$, $y>a$, $y=a$;
\item Processing of $\mathtt{min}$ and  $\mathtt{max}$;
\item Defining equality and inequalities;
\item Folding of the domain;
\item Deleting $\mathtt{max}$ in the initialization clauses;
\item From initialization clauses to input clauses;
\item Elimination of atoms $R(x,y)$ as hypotheses.
\end{enumerate}
In each of these 10 steps, we will introduce \emph{new} (binary) \emph{computation predicates}, to be added to the set $\mathbf{R}$ of existentially quantified predicates, and new clauses to define them.

\medskip
\noindent
{\bf 1) Processing the contradiction clauses:} Without loss of generality, one can assume there is the \emph{only} contradiction clause $\mathtt{max}(x)\land \mathtt{max}(y)\land R_{\bot}(x,y)\to\bot$. Indeed, each contradiction clause $\ell_1\land\ldots\land\ell_k\to\bot$ can be equivalently replaced by the conjunction of computation clauses $\ell_1\land\ldots\land\ell_k\to R_{\bot}(x,y)$ with the clause $R_{\bot}(x,y)\to\bot$ where $R_{\bot}$ is a new computation predicate (intuitively, always false). However, in place of the previous clause, we ``delay'' the contradiction, by propagating predicate $R_{\bot}$ till point $(n,n)$, thanks to the conjunction of the ``transport'' clauses $R_{\bot}(x-1,y) \land \lnot \mathtt{min}(x)\to R_{\bot}(x,y)$ and $R_{\bot}(x,y-1) \land \lnot \mathtt{min}(y)\to R_{\bot}(x,y)$ and of the unique contradiction clause  $\mathtt{max}(x) \land \mathtt{max}(y)\land R_{\bot}(x,y)\to \bot$ required by the normal form.

\medskip
\noindent
{\bf 2) Processing the input:} The idea is to make available the letters of the input word \emph{only} on the sides $x=1$ and $y=1$ of the square $\{(x,y)\in[1,n]^2\}$, this by carrying out their transport thanks to new ``transport'' predicates $W^x_s$ and $W^y_s$, for $s\in\Sigma$, inductively defined by the following clauses:
\\
$\bullet$ \emph{initialization clauses} $Q_s(x)\land \mathtt{min}(y) \to W^x_s(x,y)$ and $Q_s(y)\land \mathtt{min}(x) \to W^y_s(x,y)$;\\ 
$\bullet$ \emph{transport clauses} $W^x_s(x,y-1)\land \lnot \mathtt{min}(y) \to W^x_s(x,y)$ and $W^y_s(x-1,y)\land \lnot\mathtt{min}(x) \to W^y_s(x,y)$.

\noindent
By transitivity, these clauses imply clauses $Q_s(x)~\to~W^x_s(x,y)$ 
 and $Q_s(y)~\to~W^y_s(x,y)$.\linebreak
In other words, 
the minimal model of the conjunction of those clauses that expands structure $\langle w \rangle$ satisfies  equivalences
$\forall x \forall y \; (W^x_s(x,y) \iff Q_s(x))$ and $\forall x \forall y \; (W^y_s(x,y) \iff Q_s(y))$.
This justifies the replacement of the input atoms of form $Q_s(x-a)$ and $Q_s(y-b)$ by the respective atoms $W^x_s(x-a,y)$ and $W^y_s(x,y-b)$ in all the clauses, except in the initialization clauses.

\medskip
\noindent
{\bf 3) Restriction of computation atoms to $R(x-1,y)$, $R(x,y-1)$, $R(x,y)$:}
The idea is to introduce new ``shift'' predicates $R^{x-a}$, $R^{y-b}$ and $R^{x-a,y-b}$, for fixed integers $a,b>0$ and $R\in\mathbf{R}$: typically, we define the predicate $R^{x-a,y-b}$ that intuitively satisfies the equivalence $R^{x-a,y-b}(x,y)\iff R(x-a,y-b)$.
Let us suggest the method by an example. Assume we have initially the Horn clause 
\begin{center}
$x>3\land y>2 \land R(x-2,y-1)\land S(x-3,y-2) \to T(x,y).$
\end{center}
This clause is replaced by the clause
\begin{center}
$x>3\land y>2 \land R^{x-2}(x,y-1)\land S^{x-2,y-2}(x-1,y) ) \to T(x,y),$
\end{center}
for which the predicates $R^{x-1}$ and $R^{x-2}$ are defined by the clauses
$x>1\land R(x-1,y)\to R^{x-1}(x,y)$ and $x>2\land R^{x-1}(x-1,y)\to R^{x-2}(x,y)$
which imply $x>2 \land R(x-2,y)\to R^{x-2}(x,y)$ and then $x>2 \land y>1 \linebreak
\land R(x-2,y-1)\to R^{x-2}(x,y-1)$,
and the predicates $S^{x-1}$, $S^{x-2}$, $S^{x-2,y-1}$ and $S^{x-2,y-2}$ defined by the respective clauses:
$x>1\land S(x-1,y)\to S^{x-1}(x,y)$,\linebreak
$x>2 \land S^{x-1}(x-1,y)\to S^{x-2}(x,y)$,\linebreak
$x>2 \land y>1 \land S^{x-2}(x,y-1)\to S^{x-2,y-1}(x,y)$, \linebreak
and $x>2 \land y>2 \land S^{x-2,y-1}(x,y-1)\to S^{x-2,y-2}(x,y)$,\linebreak
which imply together the clause\\
$x>2 \land y>2 \land S(x-2,y-2) \to S^{x-2,y-2}(x,y)$
and then also
$x>3 \land y>2 \land S(x-3,y-2) \to S^{x-2,y-2}(x-1,y)$.

\medskip
\begin{remark}
Atoms on $\mathtt{min}$ and $x$ are of the forms $\mathtt{min}(x-a)$ or $\lnot\mathtt{min}(x-a)$ for $a\ge 0$, or, equivalently, $x=a+1$ or $x>a+1$. Besides, for each integer $a\ge 1$, the atom $\mathtt{max}(x-a)$ is false.
Therefore, one may consider that the only literals on $x$ involving $\mathtt{min}$ or $\mathtt{max}$ are of the form 
 $\mathtt{min}(x)$, $\lnot\mathtt{min}(x)$, $\mathtt{max}(x)$, $\lnot\mathtt{max}(x)$, $x=a$, $x>a$, for an integer $a>1$, and similarly, for $y$.
\end{remark}

\smallskip
\noindent
{\bf 4)  Elimination of atoms $x>a$, $x=a$, $y>a$, $y=a$:}
By recurrence on integer $a\ge 1$, let us define the \emph{binary} predicates $R^{x>a}$ (and, similarly, $R^{x=a}$, $R^{y>a}$, $R^{y=a}$)
whose intuitive meaning is  $x>a$ (resp. $x=a$, $y>a$, $y=a$).
The predicate $R^{x>1}$ is defined by the clause $\lnot\mathtt{min}(x) \to R^{x>1}(x,y)$. 
For $a>1$, let us define $R^{x>a}$ from $R^{x>a-1}$ by the clause
$ 
R^{x>a-1}(x-1,y) \land \lnot\mathtt{min}(x) \to R^{x>a}(x,y).
$
By recurrence on integer $a\ge 1$, these clauses imply
$
x>a\to R^{x>a}(x,y).
$
\linebreak
This justifies the replacement of the atoms $x>a$ and 
\linebreak $x=a$, for $a>1$, by $R^{x>a}(x,y)$ and $R^{x=a}(x,y)$, 
respectively, and similarly for $y$ in place of $x$. 

\medskip
After step 4, the only literals involving $\mathtt{min}$ or $\mathtt{max}$ are 
$(\lnot)\mathtt{min}(x)$, $(\lnot)\mathtt{max}(x)$, $(\lnot)\mathtt{min}(y)$, $(\lnot)\mathtt{max}(y)$.

\medskip
\noindent
{\bf 5) Processing of $\mathtt{min}$ and  $\mathtt{max}$:} 
To each literal $\eta(x)$ 
of the form 
$\mathtt{min}(x)$, $\lnot\mathtt{min}(x)$, $\mathtt{max}(x)$ or $\lnot\mathtt{max}(x)$, 
associate the new binary predicate $R^{\eta(x)}$ 
defined by the conjunction of the \emph{initialization clause}
$\eta(x)\land \mathtt{min}(y) \to R^{\eta(x)}(x,y)$ 
and  of the \emph{transport clause} $R^{\eta(x)}(x,y-1) \land \lnot \mathtt{min}(y) \to R^{\eta(x)}(x,y)$.
Do similarly for the literals $\eta(y)\in \{(\lnot) \mathtt{min}(y), (\lnot) \mathtt{max}(y)\}$.
This justifies we replace each such literal $\eta(x)$ (resp. $\eta(y)$) by the ``equivalent'' atom $R^{\eta(x)}(x,y)$ 
(resp. $R^{\eta(y)}(x,y)$) in all the clauses, except in the above initialization clauses and in the contradiction clause or in case $\eta(x)$ (resp. $\eta(y)$) is $\lnot \mathtt{min}(x)$ (resp. $\lnot \mathtt{min}(y)$) and is joined to a hypothesis of the form $R(x-1,y)$ (resp. $R(x,y-1)$).

\medskip
\noindent
{\bf Recapitulation:} After step~5 each clause is of one of the following forms:
\begin{enumerate}
\item an \emph{initialization clause} of one of the two forms:\\
$\bullet$ \emph{initialization} for $x=1$: $\mathtt{min}(x) \land \eta(y)\to R(x,y)$\\ 
with $\eta(y) \in \{(Q_s(y))_{s\in \Sigma},
(\lnot)\mathtt{min}(y),(\lnot)\mathtt{max}(y)\}$;\\
$\bullet$ \emph{initialization} for $y=1$: $\mathtt{min}(y) \land \eta(x)\to R(x,y)$\\ 
with $\eta(x) \in \{(Q_s(x))_{s\in \Sigma},
(\lnot)\mathtt{min}(x),(\lnot)\mathtt{max}(x)\}$;\
\item ``the'' \emph{contradiction clause}
$\mathtt{max}(x)\land \mathtt{max}(y)\land R_{\bot}(x,y)\to \bot$;
\item a \emph{computation clause} of the form\\
$\delta_1(x,y)\land\ldots \land \delta_r(x,y)\to R(x,y)$,
where each hypothesis $\delta_i$ is of one of the three forms
$R(x,y)$, $R(x-1,y)\land \lnot\mathtt{min}(x)$,  $R(x,y-1)\land \lnot\mathtt{min}(y)$.\\
In fact, without loss of generality, we can \emph{assume}
that each computation clause is of one of the following forms:
\begin{enumerate}
\item  $S(x-1,y)\land  \lnot \mathtt{min}(x) \to R(x,y)$;
\item $S(x,y-1)\land  \lnot \mathtt{min}(y) \to R(x,y)$;
\item $S(x,y)\land T(x,y) \to R(x,y)$.
\end{enumerate}

\end{enumerate}

\medskip
\noindent
\emph{Justification of the assumption:} ``Decompose'' each computation clause into clauses of forms (a,b,c) by introducing new intermediate predicates. For example, the computation clause
$R_1(x-1,y)\land  \lnot \mathtt{min}(x)\land R_2(x,y-1)\land \lnot \mathtt{min}(y)\land R_3(x,y) 
\to R_4(x,y)$
is ``equivalent'' to the conjunction of the following clauses using new predicates $R_5,R_6,R_7$:
$R_1(x-1,y)\land  \lnot \mathtt{min}(x)\to R_5(x,y)$; 
$R_2(x,y-1)\land  \lnot \mathtt{min}(y)\to R_6(x,y)$;
$R_5(x,y)\land R_6(x,y)\to R_7(x,y)$;
$R_7(x,y)\land R_3(x,y) \to R_4(x,y)$.

\medskip
We now plan to fold the square domain $\{(x,y)\in[1,n]^2\}$ 
along the diagonal $x=y$ on the \emph{over-diagonal triangle} $T_n=\{(x,y)\in[1,n]^2 \vert\ x\le y\}$. This requires to first define equality and inequalities.

\medskip
\noindent
{\bf 6) Defining equality and inequalities:} Let us jointly define the predicates $R_=$ and $R_{\mathtt{pred}}$ of intuitive meaning $R_=(x,y) \iff x=y$ and $R_{\mathtt{pred}}(x,y)\iff x-1=y$\linebreak
by the following clauses:
$\mathtt{min}(x)\land \mathtt{min}(y) \to R_=(x,y)$; $\lnot \mathtt{min}(x) \land R_=(x-1,y) \to R_{\mathtt{pred}}(x,y)$; 
$\lnot \mathtt{min}(y) \land R_{\mathtt{pred}}(x,y-1)\to R_=(x,y)$.

Then define the predicate $R_<$ such that $R_<(x,y) \iff x<y$ with the two clauses
$\lnot \mathtt{min}(y) \land R_=(x,y-1) \to R_<(x,y)$ and $\lnot \mathtt{min}(y) \land R_<(x,y-1) \to R_<(x,y)$.
Define similarly the predicate $R_{\le}$ such that $R_{\le}(x,y) \iff x\le y$\linebreak
 with the two clauses
$\mathtt{min}(x) \land \mathtt{min}(y) \to R_{\le}(x,y)$ and $\lnot \mathtt{min}(x) \land R_{<}(x-1,y)\to R_{\le}(x,y)$.

For easy reading, we will freely write $x=y$, $x< y$ and $x\le y$ in place of the atoms $R_{=}(x,y)$, $R_{<}(x,y)$ and $R_{\le}(x,y)$, respectively.

\medskip
\noindent
{\bf 7) Folding of the domain:}  Let us fold the square domain $\{(x,y)\in[1,n]^2\}$  
along the diagonal $x=y$ on the over-diagonal triangle $T_n=\{(x,y)\in[1,n]^2 \vert\ x\le y\}$
so that each point $(y,x)$ such that $x\le y$ is sent to its symmetrical point $(x,y)\in T_n$.
For that purpose, let us associate to each predicate $R\in\mathbf{R}$ a new (inverse) predicate $R^{\mathtt{inv}}$ 
whose intuitive meaning is the following: for each $x\le y$, we have $R^{\mathtt{inv}}(x,y)\iff R(y,x)$. So, each clause $C$ will be replaced by two clauses: the first one is the restriction of $C$ to the triangle $T_n$; the second one is the \emph{folding} on $T_n$ of the restriction of $C$ to the \emph{under-diagonal triangle} using predicates $R^{\mathtt{inv}}$. Finally, we will express that each $R\in\mathbf{R}$ coincides with its inverse $R^{\mathtt{inv}}$ on the fold $x=y$.

\smallskip
\emph{Folding the initialization clauses:} 
Each initialization clause of the form $\mathtt{min}(x) \land \eta(y)\to R(x,y)$ 
(with $\eta(y) \in \{Q_s(y) \vert s\in \Sigma\}\cup  \{(\lnot)\mathtt{min}(y),(\lnot)\mathtt{max}(y)\}$) applies to the line $x=1$ which is included in the triangle $T_n$ and consequently it should be \emph{unchanged} in the folding; in contrast, each initialization clause of the form $\mathtt{min}(y) \land \eta(x)\to R(x,y)$ (with $\eta(x) \in \{Q_s(x) \vert s~\in~\Sigma\}\cup  \{(\lnot)\mathtt{min}(x),(\lnot)\mathtt{max}(x)\}$) is \emph{replaced} by its \emph{folded} version $\mathtt{min}(x) \land \eta(y)\to R^{\mathtt{inv}}(x,y)$.

\smallskip
\emph{Folding the computation clauses:}
Let us describe how to fold the clauses (a) or (b) (folding clauses (c) is easy): \\
$\bullet$ \emph{Folding of clauses} (a):  A clause of the form $S(x-1,y)\land  \lnot \mathtt{min}(x) \to R(x,y)$ is equivalent to the conjunction of clauses
i) $x\le y \land S(x-1,y)\land  \lnot \mathtt{min}(x) \to R(x,y)$ and\\
ii) $x> y \land S(x-1,y)\land  \lnot \mathtt{min}(x) \to R(x,y)$.
Notice that clause (i) applies to the triangle $T_n$ since $x\le y$ implies $x-1<y$: therefore, clause (i) should be left \emph{unchanged}. Clause~(ii) is equivalent (by exchanging variables $x$ and $y$) to the clause
$y>x \land S(y-1,x)\land \lnot \mathtt{min}(y) \to R(y,x)$ 
whose folded (equivalent) form on $T_n$ is $x<y \land S^{\mathtt{inv}}(x,y-1)\land \lnot \mathtt{min}(y) \to R^{\mathtt{inv}}(x,y)$ since $x<y$ implies $x\le y-1$.\\
$\bullet$ \emph{Folding of clauses} (b): Similarly, a clause of the form $S(x,y-1)\land \lnot \mathtt{min}(y) \to R(x,y)$ is equivalent to the conjunction of clauses $x<y \land S(x,y-1)\land \lnot \mathtt{min}(y) \to R(x,y)$
and $x\le y \land S^{\mathtt{inv}}(x-1,y)\land \lnot \mathtt{min}(x)\to R^{\mathtt{inv}}(x,y)$.

\smallskip
\emph{Folding the contradiction clause:} Clearly, it is harmless to confuse the (contradiction) predicate
$R_{\bot}$ and its inverse $(R_{\bot})^{\mathtt{inv}}$; consequently, the contradiction clause itself
$\mathtt{max}(x) \land \mathtt{max}(y)\land R_{\bot}(x,y)\to \bot$ is its own folded version.

\smallskip
\emph{The diagonal fold:} Finally, for each $R\in\mathbf{R}$, the following two clauses mean that $R$ coincides with its inverse $R^{\mathtt{inv}}$ on the diagonal: 
$x=y\land R(x,y) \to R^{\mathtt{inv}}(x,y)$; $x=y\land R^{\mathtt{inv}}(x,y) \to R(x,y)$.

\medskip
\noindent
{\bf Recapitulation:} By a careful examination of the set of clauses obtained after steps 1-7, we can check that each of them is of one of the following forms:
\begin{enumerate}
\item an \emph{initialization clause} of the form:
$\mathtt{min}(x) \land \eta(y)\to R(x,y)$ 
with $\eta(y) \in \{Q_s(y) \vert s\in \Sigma\}\cup  \{(\lnot)\mathtt{min}(y),(\lnot)\mathtt{max}(y)\}$;
\item ``the'' \emph{contradiction clause} $\mathtt{max}(x)\land \mathtt{max}(y)\land R_{\bot}(x,y)\to \bot$;
\item a \emph{computation clause} of one of the following forms:\\
(a) $x\le y \land S(x-1,y)\land  \lnot \mathtt{min}(x) \to R(x,y)$;\\
(b) $x<y \land S(x,y-1)\land  \lnot \mathtt{min}(y) \to R(x,y)$;\\
(c) $x\le y \land S(x,y)\land T(x,y) \to R(x,y)$;\\
(d) $x=y \land S(x,y) \to R(x,y)$.
\end{enumerate}

\noindent
{\bf 8) Deleting $\mathtt{max}$ in the initialization clauses:}
The idea is to consider in parallel for each point $(x,y)$ the case where the hypothesis $\mathtt{max}(y)$ holds and the contrary case where the negation $\lnot\mathtt{max}(y)$ holds. For that purpose, we duplicate each computation predicate  
$R$ in two new predicates denoted $R^y_{\hyp \mathtt{max}}$ and $R^y_{\hyp \lnot\mathtt{max}}$. Intuitively, the atom $R^y_{\hyp \mathtt{max}}(x,y)$ 
(resp. $R^y_{\hyp \lnot\mathtt{max}}(x,y)$)
expresses the implication $\mathtt{max}(y)\to R(x,y)$ (resp. $\lnot\mathtt{max}(y)\to R(x,y)$).

\emph{Transforming the initialization clauses:}
According to the desired semantics of $R^y_{\hyp \mathtt{max}}$ 
and $R^y_{\hyp \lnot\mathtt{max}}$, each initialization clause of the form $\mathtt{min}(x) \land \mathtt{max}(y)\to R(x,y)$ (resp. $\mathtt{min}(x) \land \lnot\mathtt{max}(y)\to R(x,y)$) should be rewritten as
$\mathtt{min}(x) \to R^y_{\hyp \mathtt{max}}(x,y)$
(resp. $\mathtt{min}(x) \to R^y_{\hyp \lnot\mathtt{max}}(x,y)$).
Similarly, each initialization clause of the form $\mathtt{min}(x) \land \eta(y)\to R(x,y)$,
for $\eta(y) \in \{Q_s(y) \vert s\in \Sigma\}\cup \{(\lnot)\mathtt{min}(y)\}$
should be replaced by the conjunction of the following two clauses:
$\mathtt{min}(x) \land \eta(y)\to R^y_{\hyp \mathtt{max}}(x,y)$
and $\mathtt{min}(x) \land \eta(y)\to R^y_{\hyp \lnot\mathtt{max}}(x,y)$.

\emph{Transforming the computation clauses:}
We describe it for each above form (a-d).\\
$\bullet$ Each clause (a) $x\le y \land S(x-1,y)\land  \lnot \mathtt{min}(x) \to R(x,y)$ is replaced by the ``equivalent'' conjunction of the following two clauses
 $x\le y \land S^y_{\hyp \mathtt{max}}(x-1,y)\land  \lnot \mathtt{min}(x) \to R^y_{\hyp \mathtt{max}}(x,y)$ and
 $x\le y \land S^y_{\hyp \lnot \mathtt{max}}(x-1,y)\land  \lnot \mathtt{min}(x) \to R^y_{\hyp \lnot\mathtt{max}}(x,y)$.\\
$\bullet$ Each clause (b) $x<y \land S(x,y-1)\land \lnot \mathtt{min}(y) \to R(x,y)$ is ``equivalent''  to 
$x<y \land S^y_{\hyp \lnot \mathtt{max}}(x,y-1)\land  \lnot \mathtt{min}(y) \to R(x,y)$ since the hypothesis $\lnot\mathtt{max}(y-1)$ always holds. Consequently, clause (b) should be replaced by the ``equivalent''  conjunction of the following two clauses:\\
$x<y \land S^y_{\hyp \lnot \mathtt{max}}(x,y-1)\land  \lnot \mathtt{min}(y) \to R^y_{\hyp \mathtt{max}}(x,y)$ and\\
$x<y \land S^y_{\hyp \lnot \mathtt{max}}(x,y-1)\land  \lnot \mathtt{min}(y) \to R^y_{\hyp \lnot\mathtt{max}}(x,y)$.\\
$\bullet$ Each clause (c) $x\le y \land S(x,y)\land T(x,y) \to R(x,y)$
is replaced by the ``equivalent''  conjunction of the following two clauses:
$x\le y \land S^y_{\hyp \mathtt{max}}(x,y)\land T^y_{\hyp \mathtt{max}}(x,y) \to R^y_{\hyp \mathtt{max}}(x,y)$ 
and $x\le y \land S^y_{\hyp \lnot \mathtt{max}}(x,y)\land T^y_{\hyp \lnot \mathtt{max}}(x,y) \to R^y_{\hyp \lnot\mathtt{max}}(x,y)$.
$\bullet$ Make a similar substitution for each above clause (d). 

\emph{Processing the contradiction clause:} Obviously, the contradiction clause
$\mathtt{max}(x) \land \mathtt{max}(y)\land R_{\bot}(x,y)\to \bot$
is equivalent to the formula
$
\mathtt{max}(x) \land \mathtt{max}(y) \land (\mathtt{max}(y)\to R_{\bot}(x,y))\to \bot
$
and should be rewritten as
$\mathtt{max}(x) \land \mathtt{max}(y)\land (R_{\bot})^y_{\hyp \mathtt{max}}(x,y)\to \bot$,
which is of the required form if the predicate $(R_{\bot})^y_{\hyp \mathtt{max}}$ is renamed $R_{\bot}$.

\medskip
\noindent
{\bf 9) From initialization clauses to input clauses:}
The \emph{initialization clauses} are now of the form $\mathtt{min}(x)\to R(x,y)$ or
$\mathtt{min}(x) \land \eta(y)\to R(x,y)$,
for $\eta(y) \in \{Q_s(y) \vert s\in \Sigma\} \linebreak \cup \{(\lnot)\mathtt{min}(y)\}$.
By a separation in cases, it is easy to transform each of these clauses into an equivalent conjunction of \emph{input clauses} of the required (normalized) forms:
 $\mathtt{min}(x) \land \mathtt{min}(y)\land Q_s(y)\to R(x,y)$;
 $\mathtt{min}(x) \land \lnot\mathtt{min}(y)\land Q_s(y)\to R(x,y)$.
 
\medskip
After step 9, the formula obtained is of the claimed normal form, \emph{except} that some computation clauses may have atoms $R(x,y)$ as hypotheses. Our last step is to eliminate such hypotheses.

\medskip
\noindent 
{\bf 10) Elimination of atoms $R(x,y)$ as hypotheses:} 
The first idea is to group together in each computation clause the hypothesis atoms of the form $R(x,y)$ and the conclusion of the clause. Accordingly, the formula can be rewritten in the form
$\Phi:= \exists \mathbf{R} \forall x \forall y [ \bigwedge_{i} C_i (x,y) \land \bigwedge_{i\in[1,k]} (\alpha_i(x,y)\to \theta_i(x,y)) ]$
where the $C_i$'s are the input clauses and the contradiction clause, and each computation clause is written in the form $\alpha_i(x,y) \to \theta_i(x,y)$ where $\alpha_i(x,y)$ is a conjunction of formulas of the only forms $R(x-1,y)\land \lnot\mathtt{min}(x)$, $R(x,y-1)\land \lnot\mathtt{min}(y)$, \emph{but not} $R(x,y)$, and $\theta_i(x,y)$ is a Horn clause whose \emph{all} the atoms are of the form $R(x,y)$.
The second idea is to ``solve'' the Horn clauses $\theta_i$ (containing  \emph{only atoms} of the form $R(x,y)$) according to the input clauses and \emph{all the possible} conjunctions of hypotheses $\alpha_i$ that may be true. Notice the two following facts: the hypotheses of the input clauses are input literals and the conjuncts of the $\alpha_i$'s are of the only forms $R(x-1,y)\land \lnot\mathtt{min}(x)$, $R(x,y-1)\land \lnot\mathtt{min}(y)$. So, we can prove by induction on the sum $x+y$ that the obtained formula $\Phi'$ which is a conjunction of clauses (whose hypotheses do include no atom of the form $R(x,y)$ anymore) is equivalent to the above formula $\Phi$.
For a detailed proof, see the appendix.

\medskip
\noindent
{\bf General case:} Steps~1-7 are easy to adapt in the general case where the initial formula may contain hypotheses of the form $R(y-b,x-a)$. The new points are the following:
step~3 restricts the computation atoms to four forms:  $R(x,y)$, $R(y,x)$, $R(x-1,y)$ and $R(x,y-1)$;
step~7 (folding the domain) is adapted so that it eliminates the atoms of the form $R(y,x)$ by using the following equivalence for $x\le y$, $R(y,x)\iff R^{\mathtt{inv}}(x,y)$.
See the details in the appendix.

\medskip
This achieves the proof of the normalization result $\mathtt{pred}$-$\mathtt{ESO}$-$\mathtt{HORN}=\mathtt{normal}$-$\mathtt{pred}$-$\mathtt{ESO}$-$\mathtt{HORN}$.

\medskip
\noindent
{\bf Adaptation of steps 1-10 for normalization of predecessor Horn formulas with diagonal input-output:}
Step 1 is not modified. In step 2, the initialization clauses are now 
 $x=y\land Q_s(x)\to W^x_s(x,y)$ and  $x=y\land Q_s(y) \to W^y_s(x,y)$ whereas the transport clauses are not modified. 
 Steps 3 to 5 and the recapitulation after step 5 are not modified either, except that now each initialization clause is of one of the three forms: 1)~$x=y\land Q_s(x)\to R(x,y)$; 
 2) $\mathtt{min}(x) \land \eta(y)\to R(x,y)$, with $\eta(y) \in \{(\lnot)\mathtt{min}(y),(\lnot)\mathtt{max}(y)\}$;
 3) $\mathtt{min}(y) \land \eta(x)\to R(x,y)$, with $\eta(x) \in \{(\lnot)\mathtt{min}(x),(\lnot)\mathtt{max}(x)\}$.
 Steps 6 and~7 (folding of the domain) and the recapitulation after step~7 are not modified either, except that now each initialization clause has one of the only forms 1 and~2 above. 
 
Step~8 (deleting $\mathtt{max}$ in the initialization clauses) can be easily adapted according to those initialization clauses 
 whose forms after step 8 are now restricted to \linebreak
 1)~$x=y\land Q_s(x)\to R(x,y)$; 
 2) $\mathtt{min}(x) \land\mathtt{min}(y)\to R(x,y)$; 
 3) $\mathtt{min}(x) \land \lnot \mathtt{min}(y)\to R(x,y)$.
 Clause 2 can be replaced by the equivalent clause 2') $x=y \land \mathtt{min}(x)\to R(x,y)$.
 
 Step~9 (from initialization clauses to input clauses) is modified as follows.
 Define the predicates $R^{\mathtt{min}(x)}$, $R^{\mathtt{min}(y)}$ and $R^{\lnot \mathtt{min}(y)}$ by the initialization clauses\linebreak
 4) $x=y \land \mathtt{min}(x)\to R^{\mathtt{min}(x)}(x,y)$ and \linebreak
 5) $x=y \land \mathtt{min}(x)\to R^{\mathtt{min}(y)}(x,y)$,
 and the computation clauses 
 $\lnot \mathtt{min}(y)\land R^{\mathtt{min}(x)}(x,y-1) \to R^{\mathtt{min}(x)}(x,y)$,
 $\lnot \mathtt{min}(y)\land R^{\mathtt{min}(y)}(x,y-1) \to R^{\lnot \mathtt{min}(y)}(x,y)$,
 and $\lnot \mathtt{min}(y)\land R^{\lnot\mathtt{min}(y)}(x,y-1) \to R^{\lnot \mathtt{min}(y)}(x,y)$.
 This allows to replace the initialization clause 3 by the computation clause
 $R^{\mathtt{min}(x)}(x,y) \land R^{\lnot \mathtt{min}(y)}(x,y)\to R(x,y)$. 
 After those transformations all the initialization clauses are of the form $x=y\land Q_s(x)\to R(x,y)$ (clause 1 above) 
 or $x=y\land \mathtt{min}(x)\to R(x,y)$ (clauses 2', 4 and 5 above).
By a separation in cases, it is easy to transform each of these clauses into an equivalent conjunction of \emph{input clauses} of the required (normalized) forms:
$x=y\land \mathtt{min}(x)\land Q_s(x) \to R(x,y)$, or
 $x=y\land \lnot\mathtt{min}(x)\land Q_s(x) \to R(x,y)$.
 
 Step 10 (Elimination of atoms $R(x,y)$ as hypotheses) and the end of the proof are the same as those for 
 $\mathtt{pred}$-$\mathtt{ESO}$-$\mathtt{HORN}$. This achieves the proof of the equality
 $\mathtt{pred}$-$\mathtt{dio}$-$\mathtt{ESO}$-$\mathtt{HORN}=\mathtt{normal}$-$\mathtt{pred}$-$\mathtt{dio}$-$\mathtt{ESO}$-$\mathtt{HORN}$.
 Lemma~\ref{lemmapred} is proved.
$\square$

\medskip
\subsubsection*{Proof of Lemma~\ref{lemmaincl}: Normalization of inclusion Horn formulas}
  
It divides into seven steps.

\medskip
\noindent
{\bf 1) Processing the contradiction clauses:}  Here again, we delay the contradiction and propagate the predicate $R_{\bot}$ till the point $(1,n)$ by the conjunction of the computation clauses $x<y \land R_{\bot}(x+1,y) \to R_{\bot}(x,y)$ and \linebreak
$x<y  \land R_{\bot}(x,y-1) \to R_{\bot}(x,y)$ and of the unique \emph{contradiction clause}
$\mathtt{min}(x) \land \mathtt{max}(y) \land R_{\bot}(x,y)\to \bot$.

\medskip
\noindent
{\bf 2) Processing the input:} We make available the letters of the input word on the only diagonal $x=y$ by introducing new predicates $W^{x+a}_{s}$ and $W^{y+a}_{s}$, for $a \in\mathbb{Z}$, whose intuitive meaning is:
$W^{x+a}_{s}(x,y)\iff Q_s(x+a) \land 1 \leq x+a \leq n$ (resp. $W^{y+a}_{s}(x,y)\iff Q_s(y+a) \land 1 \leq y+a \leq n$). They are inductively defined by the following clauses:\\
$\bullet$ \emph{Initialization clauses} (on the diagonal):\\
$x=y \land Q_s(x) \to W_{s}^x(x,y)$; $(x=y) \land Q_s(x) \to W_{s}^y(x,y)$;\\
 $x=y  \land Q_s(x-a) \land x>a \to W_{s}^{x-a}(x,y)$, and \\
 $x=y \land Q_s(x-a) \land x>a \to W_{s}^{y-a}(x,y)$, for $a>0$;\\
$x=y \land Q_s(y+a) \land y\leq n-a  \to W_{s}^{x+a}(x,y)$, and \\
$x=y \land Q_s(y+a) \land y\leq n-a \to W_{s}^{y+a}(x,y)$, for $a>0$.\\
$\bullet$ \emph{Transport clauses}, for $a\in \mathbb{Z}$:\\
$x<y \land W^{x+a}_{s}(x,y-1) \to W^{x+a}_{s}(x,y)$, and\\
$x<y \land W^{y+a}_{s}(x+1,y) \to W^{y+a}_{s}(x,y)$.

This allows to replace each input atom $Q_s(x+a)$ (resp. $Q_s(y+a)$),  $a\in \mathbb{Z}$, by the computation atom 
$W^{x+a}_{s}(x,y)$ (resp. $W^{y+a}_{s}(x,y)$), in all the clauses, except in the initialization clauses.

\medskip
Note that after step 2 the atoms on input predicates $Q_s$ occur (see the initialization clauses above) always jointly with $x=y$ and in the only three forms $Q_s(x)$, $Q_s(x-a)$ (jointly with $x>a$), or $Q(y+a)$ (jointly with $y\le n-a$), for $a>0$.  

 \medskip
\noindent
{\bf 3) Processing the $\mathtt{min} / \mathtt{max}$ literals:}
One may consider that the only literals on $x$ involving $\mathtt{min}$ or $\mathtt{max}$ are of the forms 
 $x=a$, $x>a$, for an integer $a\ge1$, or $x=n-a$, $x<n-a$, for $a\ge 0$, and similarly for $y$.


As we have done for the $Q_s$, we want to make available the information about $\mathtt{min}$ and $\mathtt{max}$, i.e., about extrema, on the only diagonal $x=y$. We introduce for that new computation predicates defined inductively: 
$R^{x=a}$, $R^{x>a}$, for $a\ge1$, and $R^{x=n-a}$, $R^{x<n-a}$, for $a\ge 0$, and similarly for $y$, with obvious intuitive meaning: for instance,  $R^{y>a}(x,y) \iff y>a$.
For example, define the predicate $R^{x=a}$
(resp. $R^{y=a}$) by the two clauses $x=y \land x=a \to R^{x=a}(x,y)$ and $x<y \land R^{x=a}(x,y-1)\to R^{x=a}(x,y)$ (resp. $x=y \land x=a \to R^{y=a}(x,y)$ and $x<y \land
R^{y=a}(x+1,y)\\
\to R^{y=a}(x,y)$). As another example, define $R^{x<n-a}$ by the clauses $x=y \land y<n-a \to R^{x<n-a}(x,y)$ and $x<y \land R^{x<n-a}(x,y-1)\to R^{x<n-a}(x,y)$.

This allows to replace the ``extremum'' atoms $x=a$, $x>a$, $x=n-a$, $x<n-a$ by the respective computation \linebreak 
atoms $R^{x=a}(x,y)$, $R^{x>a}(x,y)$, $R^{x=n-a}(x,y)$, $R^{x<n-a}(x,y)$, in all the clauses, except in the initialization clauses and in the contradiction clause. And similarly for $y$.

\medskip
The important fact is that after step~3, the predicate $\mathtt{min}$ (resp. $\mathtt{max}$) only occurs in the form $x=a$ or $x>a$ (resp. in the form $y=n-a$ or $y<n-a$) and \emph{always occurs jointly with} $x=y$, i.e., is only used on the diagonal.
 
 \medskip
\noindent
{\bf 4) Restriction of computation atoms to $R(x+1,y)$, $R(x,y-1)$, and $R(x,y)$:} This is a variant of the similar step in the normalization of predecessor logics (step 3). We introduce now new ``shift'' predicates $R^{x+a}$, $R^{y-b}$ and $R^{x+a,y-b}$, for fixed integers $a,b>0$ and $R~\in~\mathbf{R}$, with easy interpretation and definitions.
In particular, the intuitive interpretation of the predicate $R^{x+a,y-b}$ is: 
$R^{x+a,y-b}(x,y) \iff x+a\le y-b \land R(x+a,y-b)$. \linebreak
As an example, the ``normalized'' clause $x<y  \land S^{x+2,y-2}(x+1,y)\to S^{x+3,y-2}(x,y)$ defines the predicate 
$S^{x+3,y-2}$ from the predicate $S^{x+2,y-2}$.

 \medskip
\noindent
{\bf Recapitulation:} After step 4, one may consider that each clause is of one of the following forms (1-3):
 \begin{enumerate}
   
\item an \emph{initialization clause} $x=y \land \delta \to R(x,y)$, where~$\delta$~is
 \begin{itemize} 
\item 
 either an input atom $Q_s(x)$, 
 \item 
 or an equality $x=a$, for a fixed $a\ge 1$, or $y=n-b$, for a fixed $b\ge 0$,
\item 
 or a conjunction $Q_s(x-a) \land x>a$, or\\ $Q_s(y+b) \land y \leq n-b$, for $a,b\ge 1$;
 \end{itemize} 
 
\item a \emph{computation clause} of one of the forms (a,b,c):
\begin{enumerate} 
\item $x<y \land S(x+1,y) \to R(x,y)$;
\item $x<y \land S(x,y-1) \to R(x,y)$;
\item $x \preceq y \land S(x,y) \land T(x,y) \to R(x,y)$, where~$\preceq\;\in \{<,=\}$;
\end{enumerate}  
 
\item ``the'' \emph{contradiction clause}\\
$\mathtt{min}(x) \land \mathtt{max}(y)\land R_{\bot}(x,y)\to \bot$,
which can be rephrased $x=1 \land y=n \land R_{\bot}(x,y)\to \bot$.
\end{enumerate}
 
\medskip
\noindent
\emph{Justification for initialization clauses:} By separation in cases, one easily obtains the above three forms of  initialization clauses.

\smallskip
\noindent
\emph{Justification for computation clauses:} Here again, ``decompose'' each computation clause in clauses of above forms (a,b,c) by introducing new intermediate predicates. For example, the computation clause
$x<y \land R_1(x+1,y)\land R_2(x,y-1) \to R_3(x,y)$
is ``equivalent'' to the conjunction of the following clauses using new predicates $R_4,R_5$:
$x<y \land R_1(x+1,y)\to R_4(x,y)$; 
$x<y \land R_2(x,y-1) \to \\ R_5(x,y)$;
$x<y \land R_4(x,y)\land R_5(x,y)\to R_3(x,y)$.

\medskip
Steps 5 and 6 that follow lie on a generalization of the method used in step 8 of the normalization of predecessor logics above (eliminating $\mathtt{max}$ in the initialization clauses). Roughly expressed, for any computation predicate
$R\in\mathbf{R}$ and a hypothesis $\eta$, we introduce a new predicate $R_{\hyp}^{\eta}$ whose intuitive meaning is:
$R_{\hyp}^{\eta}(x,y) \iff (\eta \to R(x,y))$.

\medskip
\noindent
{\bf 5) Elimination of equalities $x=a$ ($a\ge 1$) and $y=n-b$ ($b\ge 0$), except in the contradiction clause:}\linebreak
Let~$A$ (resp.~B) be the maximum of the integers $a$ (resp. $b$) that occur in the equalities $x=a$ (resp. $y=n-b$) of the clauses.
For each $R\in\mathbf{R}$, we introduce the new predicates  $R_{\hyp}^{x=a}$,  $R_{\hyp}^{y=n-b}$ and 
 $R_{\hyp}^{x=a,y=n-b}$, for all $a\in[1,A]$ and $b\in [0,B]$, whose intuitive meaning has been announced. For example, we should have $R_{\hyp}^{x=a,y=n-b}(x,y) \iff (x~=~a \land y=n-b \to R(x,y))$.
 
 \emph{Transforming the initialization clauses:} Each initialization clause $x=y \land x=a \to R(x,y)$
 (resp. \linebreak
  $x=y \land y=n-b \to R(x,y)$) is transformed into the clause $x=y \to R_{\hyp}^{x=a}(x,y)$
 (resp. $x=y \to R_{\hyp}^{y=n-b}(x,y)$).
 
  \emph{Transforming the computation clauses:} To each clause (a) $x<y \land S(x+1,y) \to R(x,y)$ add the clauses
  \linebreak
  $x<y \land S_{\hyp}^{x=a,y=n-b}(x+1,y) \to R_{\hyp}^{x=a-1,y=n-b}(x,y)$, for all $a\in[2,A]$ and $b\in [0,B]$ (\emph{justification}: the hypothesis $x+1=a$ is equivalent to $x=a-1$).
  Similarly, for each clause (b) $x<y \land S(x,y-1) \to R(x,y)$ add the clauses\linebreak
  $x<y \land S_{\hyp}^{x=a,y=n-b}(x,y-1) \to R_{\hyp}^{x=a,y=n-(b-1)}(x,y)$, for all $a\in[1,A]$ and $b\in [1,B]$.
  Also add for clauses (a,b) the similar (simplified) clauses with only one (instead of two) equality hypothesis. 
  For example, for clause (a) we add the clauses 
  $x<y \land S_{\hyp}^{x=a}(x+1,y) \to R_{\hyp}^{x=a-1}(x,y)$, for all $a\in[2,A]$, 
  and $x<y \land S_{\hyp}^{y=n-b}(x+1,y) \to R_{\hyp}^{y=n-b}(x,y)$, for all $b\in [0,B]$. 
  
  For each clause (c) $x \preceq y \land S(x,y) \land T(x,y) \to R(x,y)$, where~$\preceq\;\in \{<,=\}$, add clauses that \emph{cumulate} the hypotheses provided they are \emph{compatible}. More precisely, for all $a\in[1,A]$ and $b\in [0,B]$ and any two compatible (possibly empty) subsets $\eta, \theta$ of the set of two hypotheses $\{x=a,y=n-b\}$, we have the clause
  $x \preceq y \land S_{\hyp}^{\eta}(x,y) \land T_{\hyp}^{\theta}(x,y) \to R_{\hyp}^{\eta\cup\theta}(x,y)$.
  For example,\\ $x \preceq y \land S_{\hyp}^{x=a}(x,y) \land T_{\hyp}^{y=n-b}(x,y) \to R_{\hyp}^{x=a,y=n-b}(x,y)$
  and $x \preceq y \land S_{\hyp}^{y=n-b}(x,y) \land T_{\hyp}^{x=a,y=n-b}(x,y) \to R_{\hyp}^{x=a,y=n-b}(x,y)$.
  
  \emph{Processing the contradiction clause:} The contradiction clause
   is equivalent to 
    $x=1 \land y=n \land (x=1 \land y~=~n \to R_{\bot}(x,y))\to \bot$. Consequently, it should be replaced by the clause
    $x=1 \land y=n \land (R_{\bot})_{\hyp}^{x=1,y=n}(x,y)\to \bot$, which is the contradiction clause required if the predicate $(R_{\bot})_{\hyp}^{x=1,y=n}$ is renamed $R_{\bot}$.
  
\medskip
\noindent
{\bf 6) Elimination of atoms $Q_s(x-a)$, $Q_s(y+b)$ ($a,b>0$):}
This step is quite similar to previous step 5: it is described in the appendix.
 
\medskip
\noindent
{\bf Recapitulation:} After step 6, all the initialization clauses are of the form $x=y \land Q_s(x) \to R(x,y)$ as required\footnote{
Note that an initialization clause of the form $x=y \to R(x,y)$ can be rewritten 
$\bigwedge_{s\in\Sigma}(x=y \land Q_s(x) \to R(x,y))$ (separation in cases).}.

\medskip
\noindent
{\bf 7) Elimination of atoms $R(x,y)$ as hypotheses:} This step is exactly similar to the corresponding step~10 of normalizing predecessor logics.

\medskip
This completes the proof of the equality
 $\mathtt{incl}$-$\mathtt{ESO}$-$\mathtt{HORN}=\mathtt{normal}$-$\mathtt{incl}$-$\mathtt{ESO}$-$\mathtt{HORN}$,
 i.e., Lemma~\ref{lemmaincl}.
 $\square$

\section{Equivalence between our logics and CA complexity classes}\label{sec:Logic_CA}

The communication scheme of real-time classes finds a natural expression in our normalized logics. The \emph{input clauses}, the only clauses using unary predicates $Q_s$, express the way the input is fed to the automaton. The \emph{computation clauses} with a computation atom $R(x,y)$, $R\in \mathbf{R}$, as their conclusion simulate the computation of the CA. Deducing or not deducing the ``false'' by contradiction clauses means rejection or acceptance. Each of our normalized logics can be described graphically on a grid, indexed by $[1,n]^2$ (see Figure~\ref{fig_grid}).

\subsection*{Encoding automata states by predicates}
The set $\mathbf{R}$ of predicates that will be used in formulas expressing the computation of an automaton $\mathcal{A}$, with $Q$ the set of its states, is $\mathbf{R}=\{R_q \mid q \in Q\}$. The intuitive meaning of this predicates is the following: $R_q(c,t)$ is true $\iff$ the cell $c$ is in the state $q$ at time $t$.

\subsection*{Encoding predicates by automata states}
Let $\Phi$ be a formula defining a language $L$, in one of our logics, of the form $\Phi = \exists \mathbf{R}\forall x \forall y \psi(x,y)$ with $\mathbf{R}=(R_1,\dots,R_m)$ and $R_1=R_{\bot}$. The set of states that will be used by an automaton $\mathcal{A}$ accepting $L$ is $Q=\{\sharp,\lambda \} \cup \{0,1\}^m$ with $\sharp$ the permanent state and $\lambda$ the quiescent state. We denote by $(q)_i$ ($i \in [1,m]$), the $i^{th}$ bit of a state $q \in \{0,1\}^m$. Intuitively, $(q)_i$ refers to the predicate $R_i$. Since the predicate $R_1=R_{\bot}$ represents the ``false'', the set of accepting states is the set of states $q\in \{0,1\}^m$ with the first bit $(q)_1$ equal to 0, that is $Q_{accept}=\{0\} \times \{0,1\}^{m-1}$.


\subsection{Logical characterization of $\mathtt{RealTime_{CA}}$}\label{subsec:CA}

In this section, we will show that the languages accepted in real-time by two-way CA with input fed in a parallel way and output read on the first cell are exactly the languages defined by a formula of $ \mathtt{pred}$-$ \mathtt{ESO}$-$\mathtt{HORN}$.

\begin{theorem}\label{theorem_CA}
  $\mathtt{RealTime_{CA}} =  \mathtt{pred}$-$\mathtt{ESO}$-$\mathtt{HORN}$
  \end{theorem}

The proof will use the following inclusion scheme:
\begin{center}
  $ \mathtt{pred}$-$\mathtt{ESO}$-$\mathtt{HORN} = \mathtt{normal}$-$\mathtt{pred}$-$\mathtt{ESO}$-$\mathtt{HORN}\subseteq \mathtt{RealTime_{OIA}}$\\
 $= \mathtt{RealTime_{CA}}  \subseteq \mathtt{pred}$-$\mathtt{ESO}$-$\mathtt{HORN} $
\end{center}

The equality $ \mathtt{pred}$-$\mathtt{ESO}$-$\mathtt{HORN} = \mathtt{normal}$-$\mathtt{pred}$-$\mathtt{ESO}$-$\mathtt{HORN}$ has already been proved in Section 3 and the other equality $ \mathtt{RealTime_{OIA}}= \mathtt{RealTime_{CA}}$ is folklore in automata theory. Two inclusions are left to be proved. 

\begin{lemma}\label{lemme_CA_pred}
$\mathtt{RealTime_{CA}}  \subseteq \mathtt{pred}$-$\mathtt{ESO}$-$\mathtt{HORN}$
\end{lemma}

\begin{proof}
We will show that for each automaton $\mathcal{A} \in \mathtt{RealTime_{CA}}$ we can build a formula $\Phi \in \mathtt{pred}$-$\mathtt{ESO}$-$\mathtt{HORN}$ such that:
  $w \in L(\mathcal{A})\iff \langle w \rangle \text{ satisfies the formula } \Phi$.

First of all, one can be easily convinced that the computation power of a CA $\mathcal{A}\in \mathtt{RealTime_{CA}}$ with neighborhood $\mathcal{N}_{\mathcal{A}}=\{-1,0,1\}$ is equal to the computation power of a OCA $\mathcal{A}'$ with neighborhood $\mathcal{N}_{\mathcal{A}'}=\{-2,-1,0\}$ running within the same computation time $n$ where $n$ is the size of the input (see Figure~\ref{NeighborhoodChange}). This transformation can be seen on the space-time diagram of $\mathcal{A}$ with the variable change: $c \mapsto c+t-1$, where $c$ is the index of the cell and $t$ the time step of the computation. This geometrical transformation does not change the computation power so that: $L(\mathcal{A}) = L(\mathcal{A'})$.


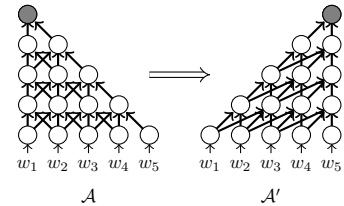
\begin{figure}[h]
\centering
\begin{tikzpicture}[scale=0.4,every node/.style={scale=0.7}]

\foreach \x in {1,2,...,4}{
\foreach \y in {\x,...,4}{
\draw[-to,shorten >=-1pt,thick](\x,5-\y) to (\x,5-\y+0.7);
}}

\foreach \x in {1,2,...,3}{
\foreach \y in {\x,...,3}{
\draw[-to,shorten >=-1pt,thick](\x,4-\y) to (\x+0.7,4-\y+0.7);
}}

\foreach \x in {1,2,...,4}{
\foreach \y in {\x,...,4}{
\draw[-to,shorten >=-1pt,thick](\x+1,5-\y) to (\x+0.3,5-\y+0.7);
}}

\foreach \x in {1,2,...,5}{
\foreach \y in {\x,...,5}{
\filldraw[fill=white] (\x,6-\y) circle (0.3cm) ;
}}

\filldraw[fill=gray] (1,5) circle (0.3cm) ;
\foreach \x in {1,2,...,5}{
\draw[->](\x,0.4) to (\x,0.7);
\draw(\x,0) node{$w_{\x}$};
}
\draw(3,-1) node{$\mathcal{A}$};

\draw[-implies,double equal sign distance] (5,3) -- (7,3);

\foreach \x in {1,2,...,4}{
\foreach \y in {\x,...,4}{
\draw[-to,shorten >=-1pt,thick](11-\x,5-\y) to (11-\x+0.7,5-\y+0.7);
}}

\foreach \x in {1,2,...,3}{
\foreach \y in {\x,...,3}{
\draw[-to,shorten >=-1pt,thick](10-\x,4-\y) to (10-\x+1.7,4-\y+0.8);
}}

\foreach \x in {1,2,...,4}{
\foreach \y in {\x,...,4}{
\draw[-to,shorten >=-1pt,thick](11-\x+1,5-\y) to (11-\x+1,5-\y+0.7);
}}

\foreach \x in {1,2,...,5}{
\foreach \y in {\x,...,5}{
\filldraw[fill=white] (12-\x,6-\y) circle (0.3cm) ;
}}

\filldraw[fill=gray] (11,5) circle (0.3cm) ;
\foreach \x in {1,2,...,5}{
\draw[->](\x+6,0.4) to (\x+6,0.7);
\draw(\x+6,0) node{$w_{\x}$};
}
\draw(9,-1) node{$\mathcal{A'}$};
\end{tikzpicture}
\caption{\label{NeighborhoodChange}Neighborhood's change: from $\{-1,0,1\}$ to $\{-2,-1,0\}$}
  \end{figure}


\emph{Input:} The parallel input of $\mathcal{A}'$ is expressed by the conjunction $\psi_{input}=\bigwedge_{s\in \Sigma} (\mathtt{min}(t) \land Q_s(c) \to R_s(c,t))$ for $\Sigma$ the input alphabet: each cell $c$ is in the state $w_c\in \Sigma$ at the start of the computation.\\[2pt]
\emph{Computation:} The state evolution of a cell of $\mathcal{A}'$ is given by the transition function:\\
$ \langle c,t \rangle = \delta_{\mathcal{A}'} (\langle c-2,t-1 \rangle,\langle c-1,t-1 \rangle,\langle c,t-1 \rangle )$. \\A transition rule $\delta_{\mathcal{A}'}(q_2,q_1,q_0)=q$ for $q_2\neq \sharp$ is expressed by the computation clause:\\
$c>t \land \lnot\mathtt{min}(t) \land R_{q_{2}}(c-2,t-1) \land R_{q_{1}}(c-1,t-1)\linebreak\land R_{q_{0}}(c,t-1)\to R_q(c,t)$.\\
\noindent
The specific case where $q_2$ is the permanent state $\sharp$ ($\delta_{\mathcal{A}'}(\sharp,q_1,q_0)=q$) is handled by the clause:\\
$c=t \land \lnot\mathtt{min}(t) \land R_{q_{1}}(c-1,t-1)\land R_{q_{0}}(c,t-1)\to R_q(c,t)$.\\
  We denote $\psi_{compute}$ the conjunction of the above two sets of clauses.\\[2pt]
  \emph{Remark:} By construction of $\psi_{input}$ and $\psi_{compute}$, we have the following equivalence for $\psi'\coloneqq \psi_{input}\land\psi_{compute}$.
  The computation atom $R_q(c,t)$ is true in the minimal model ($\langle w\rangle,\mathbf{R}$) of $\forall c \forall t  \psi'(c,t)$
  $\iff$ the cell $c$ is in the state $q$ at time $t$.\\[2pt]
  \emph{Output:} The contradiction clauses express the output on the last cell:\\
$\psi_{output} \coloneqq \bigwedge_{q\in Q\setminus Q_{accept}}(\mathtt{max}(c) \land \mathtt{max}(t) \land R_q(c,t) \to \bot)$

  The formula $\psi$ expressing the computation of $\mathcal{A}'$ is the conjunction  $\psi \coloneqq \psi'\land \psi_{output}$.\\

  \fbox{\begin{minipage}{0.9\linewidth}
      For each word $w=w_1\dots w_n \in \Sigma^+$ we have the following equivalences:

      \smallskip
  \begin{center}   
The cell $1$ of $\mathcal{A}$ is in an accepting state $q$ at time $n$\\
$\iff$\\
The cell $n$ of $\mathcal{A}'$ is in an accepting state $q$ at time $n$\\
$\iff$\\
The conjunction $\bigwedge_{q\in Q\setminus Q_{accept}}R_q(n,n)$ is false 
in the minimal model $(\langle w \rangle, \mathbf{R})$ of $\forall c \forall t \psi'(c,t)$\\
$\iff$\\
$\langle w \rangle$ satisfies the formula  $\exists\mathbf{R}\forall c \forall t \psi(c,t)$.
\end{center}
\end{minipage}
}
\smallskip

This proves $L(\mathcal{A}) \in \mathtt{pred}$-$\mathtt{ESO}$-$\mathtt{HORN}$.
\end{proof}
\smallskip

\begin{lemma}\label{lemme_pred_CA}
  $ \mathtt{normal}$-$\mathtt{pred}$-$\mathtt{ESO}$-$\mathtt{HORN}\subseteq \mathtt{RealTime_{OIA}}$
\end{lemma}

\begin{proof}
We will show that for every language $L \subseteq \Sigma ^+$ defined by a formula $\Phi \in  \mathtt{normal}$-$\mathtt{pred}$-$\mathtt{ESO}$-$\mathtt{HORN}$, a one-way iterative array $\mathcal{A} \in \mathtt{RealTime_{OIA}}$ can be constructed such that $L = L(\mathcal{A})$.

\smallskip
Let $\Phi \in \mathtt{normal}$-$\mathtt{pred}$-$\mathtt{ESO}$-$\mathtt{HORN}$ be a formula of the form $\Phi = \exists \mathbf{R}\forall x \forall y \psi(x,y)$ where $\mathbf{R}=(R_1,\dots,R_m)$ with $R_1=R_{\bot}$ and $\psi$ is a conjunction of Horn clauses of the three following forms: 
\begin{itemize}

\item \emph{input clause} of either form: 
\begin{center} 
 $\mathtt{min}(x) \land \mathtt{min}(y)\land Q_s(y)\to R(x,y)$ or
 $\mathtt{min}(x) \land \lnot\mathtt{min}(y)\land Q_s(y)\to R(x,y)$,\\
 \end{center} 

\item the \emph{contradiction clause} 
$\mathtt{max}(x)\land \mathtt{max}(y)\land R_{\bot}(x,y)\to \bot$;

\item  \emph{computation clause} of one of the following forms for some sets $H,H' \subseteq[1,m]$ and $i \in[1,m]$:
  \begin{itemize}
  \item $\lnot\mathtt{min}(x)\land\bigwedge_{h\in H}R_h(x-1,y) \land \lnot\mathtt{min}(y) \land \bigwedge_{h\in H'}R_h(x,y-1) \to R_i(x,y)$;
  \item $\lnot\mathtt{min}(x)\land\bigwedge_{h\in H}R_h(x-1,y) \to R_i(x,y)$;
    \item $\lnot\mathtt{min}(y) \land \bigwedge_{h\in H}R_h(x,y-1) \to R_i(x,y)$.
    \end{itemize}
  
\end{itemize}

 We will denote $\psi'$ the formula $\psi$ without the contradiction clause.

 We first translate the formula $\Phi$ into a dependency graph of type $\mathtt{GRID}_1$ (see Figure~\ref{fig_grid}, with input exclusively on the left side). It will then be turned into a real-time OIA $\mathcal{A}$.  
The transition function is composed by an input transition function $\delta_{\mathtt{input}}$ only applied on the first cell and the general transition function $\delta$ applied on the other cells $c$.\\[2pt] 

\newpage
\noindent For all $i\in[1,m]$ ; $s\in \Sigma$ ; $q,l,r\in Q\setminus \{\sharp,\lambda\}$:\\

\emph{Input transition function:} $\delta_{\mathtt{input}}: \Sigma \times Q \to Q $. The bit $(\delta_{input}(s,\sharp))_i$ is $1$ if and only if there is an input clause  $\mathtt{min}(x) \land \mathtt{min}(y) \land Q_s(y) \to R_i(x,y)$ in $\psi$.
 \smallskip
$(\delta_{input}(s,q))_i$ is $1$  if and only if there exists in $\psi$ an input clause $\mathtt{min}(x) \land \lnot\mathtt{min}(y) \land Q_s(y) \to R_i(x,y)$ or a computation clause $\lnot\mathtt{min}(y)\land \bigwedge_{h\in H}R_h(x,y-1) \to R_i(x,y)$
 with $(q)_h=1$, for all $h\in H$.

 \smallskip
 \emph{Transition function:}  $\delta:Q\times Q \to Q$ applied on all cells $c\in[2,n]$. The bit $(\delta(l,r))_i$ is $1$
   if and only if there  exists in $\psi$ a computation clause $\bigwedge_{h\in H}R_h(x-1,y)\land \lnot\mathtt{min}(x)\land\bigwedge_{h\in H'}R_h(x,y-1)\land \lnot\mathtt{min}(y) \to R_i(x,y)$ such that $(l)_h=1$ for all $h\in H$ and  $(r)_h=1$ for all $h\in H'$.\\
 The bit $(\delta(l,\lambda))_i$ is $1$ if and only if there is in $\psi$ a computation clause $\lnot\mathtt{min}(x)\land\bigwedge_{h\in H}R_h(x-1,y)\to R_i(x,y)$ such that $(l)_h=1$, for all $h\in H$.


\begin{figure}
\centering
\begin{tikzpicture}[scale=0.4,every node/.style={scale=0.7}]

\draw[dashed,ultra thin](1,1) grid (5,9);

\foreach \x in {1,2,...,5}{
\foreach \y in {1,...,4}{
\draw[-to,shorten >=-1pt,thick](\x,\y) to (\x,\y+0.7);
}}

\foreach \x in {1,2,...,4}{
\foreach \y in {1,...,5}{
\draw[-to,shorten >=-1pt,thick](\x,\y) to (\x+0.7,\y);
}}

\foreach \x in {1,2,...,5}{
\foreach \y in {1,...,5}{
\filldraw[fill=white] (\x,\y) circle (0.3cm) ;
}}

\filldraw[fill=gray] (5,5) circle (0.3cm) ;
\foreach \x in {1,2,...,5}{
\draw[->](0.4,\x) to (0.7,\x);
\draw(0,\x) node{$w_{\x}$};
}
\draw(3,-1) node{$(x,y)$};

\draw[-implies,double equal sign distance] (5.5,4) -- (7,4);

\draw[dashed,ultra thin](9,1) grid (13,9);

\foreach \x in {1,2,...,5}{
\foreach \y in {1,...,4}{
\draw[-to,shorten >=-1pt,thick](\x+8,\x+\y-1) to (\x+8,\y+\x-1+0.7);
}}

\foreach \x in {1,2,...,4}{
\foreach \y in {0,...,4}{
\draw[-to,shorten >=-1pt,thick](\x+8,\y+\x) to (\x+8+0.7,\y+\x+0.7);
}}

\foreach \x in {1,2,...,5}{
\foreach \y in {1,...,5}{
\filldraw[fill=white] (\x+8,\x+\y-1) circle (0.3cm) ;
}}

\filldraw[fill=gray] (13,9) circle (0.3cm) ;
\foreach \x in {1,2,...,5}{
\draw[->](8+0.4,\x) to (8+0.7,\x);
\draw(8+0,\x) node{$w_{\x}$};
}

\draw(6,-1) node{$\mapsto$};

\draw(11,-1) node{$(c=x,t=x+y-1)$};

\draw(3,10) node{$\mathtt{GRID_1}$};
\draw(11,10) node{$\mathtt{RealTime_{OIA}}$};

\end{tikzpicture}
\caption{Variable change}\label{fig:VarChangeOIA}
  \end{figure}
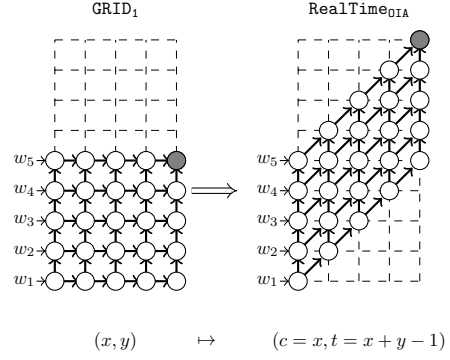


  Let $\mathcal{A}=(Q,\Sigma,Q_{accept},\mathcal{N},\delta_{\mathtt{input}}, \delta)$ be the OIA defined above after making the change of variables of Figure~\ref{fig:VarChangeOIA}:\linebreak {$c=x$;} $t=x+y-1$.
  By construction of $\mathcal{A}$, we have the following equivalences concluding the proof.\\
$ \forall w\in\Sigma^n$:
\begin{enumerate}
 \item $ \forall (a,b)\in[1,n]^2,  \forall i\in[1,m]$, the atom $R_i(a,b)$ is true in the minimal model 
$(\langle w \rangle, \mathbf{R})$ of $\forall x \forall y \; \psi'(x,y)$ $\iff$  The cell $c=a$ is at time $t=a+b-1$ in a state $q$ with $(q)_i=1$.
\item  $\mathcal{A}$ accepts $w$ in real-time 
$\iff$  $\langle w \rangle$ satisfies $\Phi$.
\end{enumerate}
\end{proof}

\subsection{Logical characterization of $\mathtt{RealTime_{IA}}$}\label{subsec:IA}

In this section, we will show that the languages accepted in real-time by IA are exactly the languages defined by formulas of $ \mathtt{pred}$-$\mathtt{dio}$-$ \mathtt{ESO}$-$\mathtt{HORN}$.

\begin{theorem}\label{theorem_IA}
$\mathtt{RealTime_{IA}} =  \mathtt{pred}$-$\mathtt{dio}$-$ \mathtt{ESO}$-$\mathtt{HORN}$.
\end{theorem}

The proof of Theorem~\ref{theorem_IA} is close to the one of Theorem~\ref{theorem_CA} and is obtained by the following inclusion scheme:\\
\smallskip
$\mathtt{pred}$-$\mathtt{dio}$-$\mathtt{ESO}$-$\mathtt{HORN} = \mathtt{normal}$-$\mathtt{pred}$-$\mathtt{dio}$-$ \mathtt{ESO}$-$\mathtt{HORN}$\\
$\subseteq \mathtt{RealTime_{IA}} \subseteq  \mathtt{pred}$-$\mathtt{dio}$-$ \mathtt{ESO}$-$\mathtt{HORN}$.\\

The equality $\mathtt{pred}$-$\mathtt{dio}$-$ \mathtt{ESO}$-$\mathtt{HORN} = \mathtt{normal}$-$ \mathtt{pred}$-$\mathtt{dio}$-$ \mathtt{ESO}$-$\mathtt{HORN}$ has been already proved in Section 3. We will now prove the two remaining inclusions.

\begin{lemma}\label{lemma_IA_dio}
$\mathtt{RealTime_{IA}} \subseteq  \mathtt{pred}$-$\mathtt{dio}$-$ \mathtt{ESO}$-$\mathtt{HORN}$.
\end{lemma}

\begin{proof}
 Let $\mathcal{A}$ be an automaton in $\mathtt{RealTime_{IA}}$, we apply the transformation  $c \mapsto c+t-1$ on its time-space diagram. This transformation gives us a new automaton $\mathcal{A'}$ with the neighborhood $\mathcal{N'}=\{-2,-1,0\}$ and the input still fed sequentially but in the following way: the $i^{th}$ bit of the input is given to the cell $i$ at time $i$.
Since the input presentation is the only change between the computation of an automaton in $\mathtt{RealTime_{IA}}$ and an automaton in $\mathtt{RealTime_{CA}}$, the computation clauses and the contradiction clauses will stay the same.

\noindent\emph{Input:} The diagonal input of $\mathcal{A'}$ is expressed by the conjunction of the input clauses:\\
\smallskip
$\psi_{input} \coloneqq\bigwedge_{s\in \Sigma} (t=c \land Q_s(c) \to R_s(c,t))$ for $\Sigma\subset Q$.

\noindent\emph{Computation:} The conjunction $\psi_{compute}$ is defined from the transition rules of $\mathcal{A'}$ as in the previous section.

\noindent Let $\psi'$ be the conjunction $\psi_{input} \land \psi_{compute}$.

\noindent\emph{Output:} The output reading is done on the same cell as in the previous section.\\
$\psi_{output}\coloneqq\bigwedge_{q\in Q\setminus Q_{accept}}(\mathtt{max}(c) \land \mathtt{max}(t) \land R_q(c,t) \to \bot)$
\smallskip

The formula $\psi$ of $ \mathtt{pred}$-$\mathtt{dio}$-$ \mathtt{ESO}$-$\mathtt{HORN}$ expressing the computation of $\mathcal{A}'$ is:  $\psi \coloneqq \psi'\land \psi_{output}$.\\

  \fbox{\begin{minipage}{0.9\linewidth}
      For each word $w=w_1\dots w_n \in \Sigma^+$ we have the following equivalences:

      \smallskip
  \begin{center}   
The cell $1$ of $\mathcal{A}$ is in an accepting state $q$ at time $n$\\
$\iff$\\
The cell $n$ of $\mathcal{A}'$ is in an accepting state $q$ at time $n$\\
$\iff$\\
The conjunction $\bigwedge_{q\in Q\setminus Q_{accept}}R_q(n,n)$ is false 
in the minimal model $(\langle w \rangle, \mathbf{R})$ of $\forall c \forall t \psi'(c,t)$\\
$\iff$\\
$\langle w \rangle$ satisfies the formula  $\exists\mathbf{R}\forall c \forall t \psi(c,t)$.
\end{center}
\end{minipage}
}\\

This proves $L(\mathcal{A}) \in\mathtt{pred}$-$\mathtt{dio}$-$ \mathtt{ESO}$-$\mathtt{HORN}$.

\end{proof}

\begin{lemma}
  $\mathtt{normal}$-$\mathtt{pred}$-$\mathtt{dio}$-$\mathtt{ESO}$-$\mathtt{HORN} \subseteq \mathtt{RealTime_{IA}}$
  \end{lemma}

  \begin{proof}[Proof sketch] Since the proof is similar to that of Lemma~\ref{lemme_pred_CA}, we will here just give an hint on how to associate to each site $(x,y)\in[1,n]^2$ such that $x\leq y$ a site $(c,t)$ of the space-time diagram of an iterative array $\mathcal{A}$ running in real-time (see Figure~\ref{dio-IA}). The sites $(x,y)\in[1,n]^2$ such that $x\geq y$ are similarly handled.\\
 First, we apply to the set of sites $(x,y)$ of the dependency graph of $\Phi$ the variable change $c'=y-x+1;t'=x+y-1$. This variable change turns respectively the dependencies $(x-1,y)\to(x,y)$ and $(x,y-1)\to (x,y)$ into $(c'+1,t'-1)\to(c',t')$ and $(c'-1,t'-1)\to(c',t')$ expressing the two-way communication of an iterative array $\mathcal{A'}$. The sites $(c',t')$ of the space-time diagram of $\mathcal{A'}$ takes their values in $[1,n]\times[1,2n-1]$ and the $i^{th}$ bit of the input is fed on the site $(1,2i-1)$ (see Figure~\ref{dio-IA}). \\
 In order to obtain the space-time diagram of an IA $\mathcal{A}$ running in real time, each site $(c,t)=(\left \lceil{c'/2}\right \rceil,\left \lceil{t'/2}\right \rceil)$ of this diagram will record the set of sites $\{(c'-1,t'-1),(c',t'),(c'+1,t'-1)\}$ of the space-time diagram of $\mathcal{A'}$, with $c'$ and $t'$ odd and greater than $1$ (see Figure~\ref{dio-IA}).\\
 

\begin{figure}[h]
\centering
\begin{tikzpicture}[scale=0.4,every node/.style={scale=0.7}]

\draw[dashed,ultra thin](1,1-1) grid (5,9-1);

\foreach \x in {1,2,...,4}{
\foreach \y in {\x,...,4}{
\draw[-to,shorten >=-1pt,thick](\x,\y-1) to (\x,\y-1+0.7);
}}

\foreach \x in {1,2,...,4}{
\foreach \y in {\x,...,4}{
\draw[-to,shorten >=-1pt,thick](\x,\y) to (\x+0.7,\y);
}}

\foreach \x in {1,2,...,5}{
\foreach \y in {\x,...,5}{
\filldraw[fill=white] (\x,\y-1) circle (0.3cm) ;
}}

\filldraw[fill=gray](5,5-1) circle (0.3cm) ;
\foreach \x in {1,2,...,5}{
   \filldraw[white,fill=white] (\x,\x-1-1) circle (0.5cm) ; 
  \draw[->](\x,\x-0.6-1) to (\x,\x-0.3-1);  
\draw(\x,\x-1-1) node{$w_{\x}$};
}
\draw(3,-1-1) node{$(x,y)$};

\draw(3,9) node{$\mathtt{GRID_3}$};

\draw[-implies,double equal sign distance] (6,4-1) -- (7.5,4-1);


\draw[dashed,ultra thin](9,1-1) grid (13,9-1);

\foreach \x in {1,2,...,4}{
  \foreach \y in {\x,...,4}{
    
    \def\c{8+\y-\x+1}
    \def\t{\y+\x}
\draw[-to,shorten >=-1pt,thick](\c+1,\t-1) to (\c+0.25,\t+0.75-1);
}}

\foreach \x in {1,2,...,4}{
  \foreach \y in {\x,...,4}{
        \def\c{8+\y-\x+1}
    \def\t{\y+\x-1}
\draw[-to,shorten >=-1pt,thick](\c,\t-1) to (\c+0.7,\t+0.7-1);
}}

\foreach \x in {1,2,...,5}{
  \foreach \y in {\x,...,5}{
    \def\c{8+\y-\x+1}
    \def\t{\y+\x-1}
\filldraw[fill=white] (\c,\t-1) circle (0.3cm) ;
}}

\filldraw[fill=gray](9,9-1) circle (0.3cm) ;
\foreach \x in {1,2,...,5}{
  \draw[->](8+0.4,2*\x-1-1) to (8+0.7,2*\x-1-1);
\draw(8+0,2*\x-1-1) node{$w_{\x}$};
}

\draw(5,-1-1) node{$\mapsto$};

\draw(10,-1-1) node{$(c'=y-x+1,t'=x+y-1)$};

\draw(11,9) node{$\mathcal{A'}$};

\draw[dashed,ultra thin,step=2](16,0) grid (20,8);

\draw[-to,shorten >=-1pt,thick](16,0) to (18-0.25,1+0.75);
\draw[-to,shorten >=-1pt,thick](18,2) to (20-0.25,3+0.75);
\draw[-to,shorten >=-1pt,thick](16,2) to (18-0.25,3+0.75);
\draw[-to,shorten >=-1pt,thick](16,4) to (18-0.25,5+0.75);

\draw[-to,shorten >=-1pt,thick](18,2) to (16+0.25,3+0.75);
\draw[-to,shorten >=-1pt,thick](18,4) to (16+0.25,5+0.75);
\draw[-to,shorten >=-1pt,thick](18,6) to (16+0.25,7+0.75);
\draw[-to,shorten >=-1pt,thick](20,4) to (18+0.25,5+0.75);

\draw[-to,shorten >=-1pt,thick](16,0) to (16,1+0.7);
\draw[-to,shorten >=-1pt,thick](16,2) to (16,3+0.7);
\draw[-to,shorten >=-1pt,thick](16,4) to (16,5+0.7);
\draw[-to,shorten >=-1pt,thick](16,6) to (16,7+0.7);
\draw[-to,shorten >=-1pt,thick](18,2) to (18,3+0.7);
\draw[-to,shorten >=-1pt,thick](18,4) to (18,5+0.7);

\filldraw[fill=white] (16,0) circle (0.3cm) ;
\filldraw[fill=white] (16,2) circle (0.3cm) ;
\filldraw[fill=white] (16,4) circle (0.3cm) ;
\filldraw[fill=white] (16,6) circle (0.3cm) ;
\filldraw[fill=white] (16,8) circle (0.3cm) ;
\filldraw[fill=white] (18,2) circle (0.3cm) ;
\filldraw[fill=white] (18,4) circle (0.3cm) ;
\filldraw[fill=white] (18,6) circle (0.3cm) ;
\filldraw[fill=white] (20,4) circle (0.3cm) ;

\filldraw[fill=gray] (16,8) circle (0.3cm) ;
\foreach \x in {1,2,...,5}{

\draw[->](15+0.4,2*\x-2) to (15+0.7,2*\x-2);
\draw(15+0,2*\x-2) node{$w_{\x}$};
}

\draw(18,-2) node{$(c,t)=(\left \lceil{\frac{c'}{2}}\right \rceil,\left \lceil{\frac{t'}{2}}\right \rceil)$};

\draw(18,9) node{$\mathcal{A}$};

\end{tikzpicture}
\caption{\label{dio-IA}Variable change and grouping}
  \end{figure}
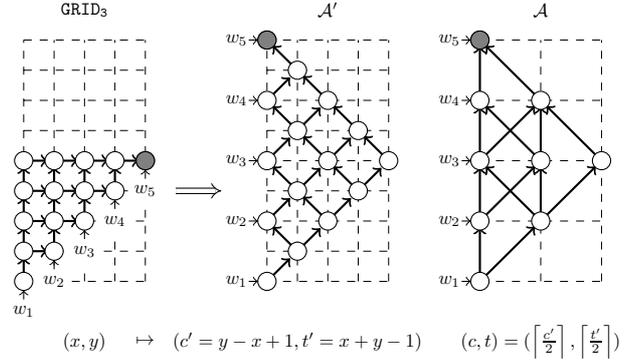
\end{proof}
  
  \subsection{Logical characterization of $\mathtt{Trellis}$ and linear conjunctive grammars}\label{subsec:Trellis-LinConj}

  Conjunctive grammars, introduced by Okhotin~\cite{Okhotin04rairo}, are an extension of context-free grammars. This class of formal grammars is obtained by allowing the use of a conjunction operator (denoted $\&$) in the right side of the context-free rules, meaning an intersection between derived languages. It has been shown in \cite{Okhotin04rairo} that languages obtained by the linear restriction of conjunctive grammars are the same as the ones recognized by trellis automata. The class of languages generated by linear conjunctive grammars is denoted $LinConj$.\\
  Each rule of a linear conjunctive grammar $\mathcal{G}=(\Sigma,N,P,S)$ in normal form is a rule of the form:\\
  \smallskip
  $A \to sB_1 \& \dots \& sB_{\ell}\& C_1t \&\dots \&C_pt$ or  $A \to s$\\
  \smallskip
  with $\ell +p \ge 1$,  $A,B_i,C_j\in N$ and $s,t\in\Sigma$;

  As a context-free language, a conjunctive language has an algebraic representation by least solution of a system of language equations with concatenation, union and intersection.
  For example, the rule \\
  \smallskip
  $A \to sB_1 \& \dots \& sB_{\ell}\&C_1t\& \dots \& C_pt$ \\
\smallskip
  gives us the language equation: \\
\smallskip
$L(A)= sL(B_1) \cap \dots \cap sL(B_{\ell})\cap L(C_1)t\cap\dots \cap L(C_p)t$ \\
\smallskip
where $L(A)$ is the set of all words having the syntactic property $A$. \\

Our inclusive logic characterizes the class of (complements of) linear conjunctive languages:

  \begin{theorem}\label{theorem_LinConj}
 For any $L \in\Sigma^+$ we have:\\
$L\in\mathtt{incl}$-$\mathtt{ESO}$-$\mathtt{HORN}$ $\iff$ $(\Sigma^+ \setminus L) \in LinConj$ .
\end{theorem}

 \noindent The theorem is proved in two steps.

\begin{lemma}\label{lemma_incl_LinConj}
$ L\in\mathtt{incl}$-$\mathtt{ESO}$-$\mathtt{HORN}\Rightarrow (\Sigma^+ \setminus L) \in LinConj$ 
\end{lemma}
\begin{proof}Let us show that for each language $L\in \mathtt{incl}$-$\mathtt{ESO}$-$\mathtt{HORN}$ defined by a formula $\Phi\in \mathtt{normal}$-$\mathtt{incl}$-$\mathtt{ESO}$-$\mathtt{HORN}$, $\Phi = \exists\mathbf{R}\forall x \forall y \psi(x,y)$, we can build a grammar $G_{\Phi}=(\Sigma,N,P,S)$ such that: $w\in L \iff w\notin L(G_{\Phi})$

\noindent The grammar $G_{\Phi}$ is defined as follows.\\
\smallskip
The set of non-terminal symbols of $G_{\Phi}$ is $N=\mathbf{R}$ and $S=R_{\bot}$ is the start symbol. The rules of $\mathcal{G}_{\Phi}$ are the following:\\
\smallskip
To each input clause $x=y \land Q_s(x)\to R(x,y)$ of $\Phi$ we associate the rule $R\to s$.\\
\smallskip
To each computation clause of $\Phi$\\
$x<y \land S_1(x+1,y)\land\ldots\land S_{\ell}(x+1,y)\land T_1(x,y-1)\land\ldots\land T_p(x,y-1)\to R(x,y)$
we associate the rules\\
 $R\to sS_1\&\ldots \&sS_{\ell}\&T_1t\&\ldots \&T_pt$, for all $s,t\in\Sigma$.

 Let $\psi'$ be the conjunction $\psi$ without the contradiction clause.
 By an easy induction, the following equivalence can be proved, for all $R \in \mathbf{R}$, all words $w=w_1\dots w_n \in \Sigma^+$, and all intervals $[a,b] \subseteq [1,n]$:\\
 $w_a \dots w_b \in L(R) \iff$ \\
 $ \left(\forall x \forall y \psi' \land \bigwedge_{i\in[a,b]} Q_{w_i}(i) \right) \to R(a,b)$ is a tautology.

 Taking $R=R_{\bot}$ and $[a,b]=[1,n]$ we obtain:\\
 $w\in L(G_{\Phi}) \iff \langle w \rangle \nvDash \Phi$ as claimed.
 
\end{proof}

\begin{lemma}\label{lemma_LinConj_incl}
$(\Sigma^+ \setminus L) \in LinConj \Rightarrow L\in\mathtt{incl}$-$\mathtt{ESO}$-$\mathtt{HORN}$ 
\end{lemma}

\begin{proof}Let $L$ be a language, subset of  $\Sigma^+$, such that $(\Sigma^+\setminus L) \in LinConj$. 
We associate to the linear grammar $\mathcal{G}=(\Sigma,N,P,S)$ defining $(\Sigma^+ \setminus L)$ with $N=\{A_1,\ldots,A_k\}$ in normal form, as presented above, the formula \linebreak
$\Phi_{\mathcal{G}}:=\exists\mathbf{R} \forall x \forall y \psi(x,y)$ with $\mathbf{R}= \{A_1,\ldots,A_k\}$ and $\psi$ the conjunction of the following Horn clauses:\\
the computation clause
$x<y \land Q_s(x) \land \bigwedge_{i=1}^{\ell} B_i(x+1,y) \land \linebreak \bigwedge_{i=1}^p C_i(x,y-1) \land Q_t(y) \to A(x,y)$
for each rule  \linebreak$A \to sB_1\& \ldots \& sB_{\ell}\&C_1t\&\ldots \& C_pt$;\\[2pt]
the clause $x=y\land Q_s(x)\to A(x,y)$ for each rule $A\to s$ ;\\[2pt]
and the contradiction clause $\mathtt{min}(x)\land \mathtt{max}(y)\land S(x,y)\to\bot$.

The proof of the equivalence $w\in L(\mathcal{G}) \iff \langle w \rangle \models \Phi_{\mathcal{G}}$ is the same as for the previous Lemma~\ref{lemma_incl_LinConj}.
\end{proof}

    \subsection*{Logical characterization of $\mathtt{Trellis}$}

Clearly, the final state (result) $q$ of a trellis automaton $\mathcal{A}=(Q,\Sigma,Q_{accept}, \delta)$ acting on a word $w_x \dots w_y$ of length greater than $1$ ($x<y$) is completely determined by the final state $q_l$ of $\mathcal{A}$ acting on the prefix $w_x \dots w_{y-1}$ and the final state $q_r$ of $\mathcal{A}$ acting on the suffix $w_{x+1} \dots w_y$: $q=\delta(q_l,q_r)$. This functional dependence $\left( (x,y-1),(x+1,y) \right) \to (x,y)$ is exactly expressed by the computation clauses of the normalized inclusive logic. This suggests the following equality.
\smallskip
\begin{theorem}\label{theorem_Trellis}
  $\mathtt{incl}$-$\mathtt{ESO}$-$\mathtt{HORN} =\mathtt{Trellis}$
  \end{theorem}

\noindent The theorem is proved in two steps.
\smallskip
 \begin{lemma}\label{lemme_Trellis_incl}
$\mathtt{Trellis} \subseteq \mathtt{incl}$-$\mathtt{ESO}$-$\mathtt{HORN}$
\end{lemma}

\begin{proof} Let $\mathcal{A}$ be a trellis automaton that accepts a language $L\subseteq \Sigma^+$ and let be a word $w=w_1\dots w_n\in\Sigma^+$. Let us introduce the set of binary predicates $\mathbf{R}=\{R_q \mid q \in Q\}$ with intuitive meaning: $R_q(x,y)$ is true $\iff$ the final state of $\mathcal{A}$ acting on the subword $w_x \dots w_y$ is $q$. The following clauses describe the computation of $\mathcal{A}$.\\
  \smallskip
  \emph{Input clauses:}\\
  \smallskip
$\psi_{input} \coloneqq \bigwedge_{s\in \Sigma} (x=y \land Q_s(x)\to R_s(x,y))$, for $s\in\Sigma$.

  \smallskip
 \noindent\emph{Computation clauses:}\\
  \smallskip
$\psi_{compute} \coloneqq \bigwedge_{(q_l,q_r) \in Q^2} \left( x<y \land R_{q_l}(x,y-1) \land \right. \linebreak \left. R_{q_r}(x+1,y) \to R_q(x,y)\right)$ where $q=\delta(q_l,q_r)$.

\smallskip
\noindent\emph{Contradiction clauses:} \\
\smallskip
$\psi_{output}\coloneqq\bigwedge_{q\in Q\setminus Q_{accept}}(\mathtt{min}(x) \land \mathtt{max}(y) \land R_q(x,y) \to \bot)$\\

\noindent
Let $\psi' \coloneqq \psi_{input} \land \psi_{compute}$ and $\psi \coloneqq \psi' \land \psi_{output}$.

  \fbox{\begin{minipage}{0.9\linewidth}
      For each word $w=w_1\dots w_n \in \Sigma^+$ we have the following equivalences:

      \smallskip
  \begin{center}   
$\mathcal{A}$ accepts $w$ in time $n$\\

$\iff$\\
The conjunction $\bigwedge_{q\in Q\setminus Q_{accept}}R_q(1,n)$ is false 
in the minimal model $(\langle w \rangle, \mathbf{R})$ of $\forall x \forall y \psi'(x,y)$\\
$\iff$\\
$\langle w \rangle$ satisfies the formula  $\exists\mathbf{R}\forall x \forall y \psi(x,y)$.
\end{center}
\end{minipage}
}\\

This proves $L(\mathcal{A}) \in \mathtt{incl}$-$\mathtt{ESO}$-$\mathtt{HORN}$.
\end{proof}

  \begin{lemma}\label{lemme_incl_Trellis}
$\mathtt{incl}$-$\mathtt{ESO}$-$\mathtt{HORN}\subseteq\mathtt{Trellis}$
\end{lemma}

\begin{proof}[Proof sketch]
  Let $\Phi$ be a formula of $\mathtt{normal}$-$\mathtt{incl}$-$\mathtt{ESO}$-$\mathtt{HORN}$, we establish a natural bijection between the sites $(x,y)$ of the domain of the formula and the sites $(c,t)$ of the space-time diagram of a OCA running in real-time (see Figure~\ref{fig_INC}). The transition function of this automaton is then deduced from the computation clauses of $\Phi$ as in sections~\ref{subsec:CA} and~\ref{subsec:IA}.

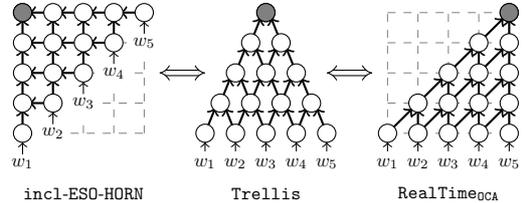
\begin{figure}[h]
\centering
\begin{tikzpicture}[scale=0.4,every node/.style={scale=0.7}]


\draw[dashed,ultra thin,gray](1,1) grid (5,5);
  
\foreach \y in {1,2,...,4}{
  \foreach \x in {\y,...,4}{
    \def\c{\y+1}
    \def\t{\x+1}
  \draw[-to,shorten >=-1pt,thick](\c,\t) to (\c-0.7,\t);
}}


\foreach \y in {1,2,...,4}{
  \foreach \x in {\y,...,4}{
    \def\c{\y}
    \def\t{\x}
  \draw[-to,shorten >=-1pt,thick](\c,\t) to (\c,\t+0.7);
}}

\foreach \y in {1,2,...,5}{
  \foreach \x in {1,...,\y}{
\filldraw[fill=white] (\x,\y) circle (0.3cm) ;
}}

\filldraw[fill=gray] (1,5) circle (0.3cm) ;
\foreach \x in {1,2,...,5}{
  \filldraw[white,fill=white] (\x,\x-1) circle (0.5cm);
  \draw[->](\x,\x-0.7) to (\x,\x-0.3);
\draw(\x,\x-1) node{$w_{\x}$};
}
\draw (3,-1) node{$\mathtt{incl}$-$\mathtt{ESO}$-$\mathtt{HORN}$};

\draw[implies-implies,double equal sign distance] (5.5,3) -- (7,3);
  

\draw[dashed,ultra thin,gray](13,1) grid (17,5);
\foreach \y in {1,2,...,4}{
  \foreach \i in {\y,...,4}{
    \def\x{\i+12}
  \draw[-to,shorten >=-1pt,thick](\x,\y) to (\x+0.75,\y+0.75);
}}


\foreach \y in {1,2,...,4}{
  \foreach \i in {\y,...,4}{
    \def\x{\i+12}
\draw[-to,shorten >=-1pt,thick](\x+1,\y) to (\x+1,\y+0.7);
}}

\foreach \y in {1,2,...,5}{
  \foreach \i in {\y,...,5}{
    \def\x{\i+12}
\filldraw[fill=white] (\x,\y) circle (0.3cm) ;
}}

\filldraw[fill=gray] (5+12,5) circle (0.3cm) ;
\foreach \i in {1,2,...,5}{
  \def\x{\i+12}
\draw[->](\x,0.4) to (\x,0.7);
\draw(\x,0) node{$w_{\i}$};
}
\draw(3+12,-1) node{$\mathtt{RealTime_{OCA}}$};

\draw[implies-implies,double equal sign distance] (11,3) -- (12.5,3);

\foreach \y in {1,2,...,4}{
  \foreach \x in {\y,...,4}{

\draw[-to,shorten >=-1pt,thick](6+0.5-0.5*\y+\x,\y) to (7+\x-0.5*\y-0.2,\y+0.75);
}}

\foreach \y in {1,2,...,4}{
\foreach \x in {\y,...,4}{

\draw[-to,shorten >=-1pt,thick](6+\x+1.5-0.5*\y,\y) to (7-0.5*\y+\x+0.2,\y+0.75);
}}

\foreach \y in {1,2,...,5}{
\foreach \x in {\y,...,5}{
\filldraw[fill=white] (6.5+\x-0.5*\y,\y) circle (0.3cm) ;
}}

\filldraw[fill=gray](6+3,5) circle (0.3cm) ;
\foreach \x in {1,2,...,5}{
\draw[->](6+\x,0.4) to (6+\x,0.7);
\draw(6+\x,0) node{$w_{\x}$};
}
\draw(6+3,-1) node{$\mathtt{Trellis}$};

\end{tikzpicture}
\caption{\label{fig_INC}The bijection between $\mathtt{incl}$-$\mathtt{ESO}$-$\mathtt{HORN}$ and $\mathtt{Trellis}$}
\end{figure}
\end{proof}

 \section{Concluding remarks}\label{sec:conc}

\noindent
It was known that the three complexity classes studied in this paper are the only distinct and natural complexity classes for minimal time, so-called real-time, of one-dimensional cellular automata. In various articles from the 1960s to 2000s, it has been established that each of those classes is robust, in particular with respect to the chosen neighborhood~\cite{Poupet05STACS}, and has several equivalent characterizations in other frameworks: e.g, $\mathtt{Trellis}$ is the class of linear conjunctive languages~\cite{Okhotin04rairo}.

In this paper, we have presented a unified view of these three real-time classes as part of descriptive complexity. We hope to have convinced the reader that the logics we have defined are useful tools to express problems in a natural way and to deduce automatically from any formula in such a logic an equivalent cellular automaton running in real-time. This is exemplified by languages $\mathtt{notBordered}$ and  $\mathtt{Palindrome}$, which can be succinctly defined in $\mathtt{pred}$-$\mathtt{ESO}$-$\mathtt{HORN}$ and $\mathtt{incl}$-$\mathtt{ESO}$-$\mathtt{HORN}$, respectively, and hence, belong to $\mathtt{RealTime}_{\mathtt{CA}}$ and $\mathtt{Trellis}$, respectively. Conversely, the problem $\mathtt{Prime}:=\{w\in \Sigma^* \;\vert\; \vert w \vert$ is a prime integer$\}$ belongs to $\mathtt{RealTime}_{\mathtt{CA}}$ (and to $\mathtt{RealTime}_{\mathtt{IA}})$ by Fischer's algorithm~\cite{Fischer65JACM, Delorme02SCA}; therefore, it belongs to $\mathtt{pred}$-$\mathtt{ESO}$-$\mathtt{HORN}$.

Further, we believe that the normalized versions of our three logics, identified to square grid circuits -- a natural concept -- offer a fresh view of the real-time complexity classes and an additional argument for their robustness.
In addition, we are thinking about broadening our logics by allowing a limited use of negation on computation atoms like in \emph{Stratified Datalog}~\cite{AbiteboulHV95}, for easier programming inside these logics and without changing their real-time complexity.

\medskip

\noindent
{\bf Acknowledgments:} 
This paper would not exist without the inspiration and guidance of Véronique Terrier. Her teaching, her deep insights into cellular automata, the references and advice she generously gave us, as well as her careful reading, were essential in designing and finalizing our concepts and results.
 
\bibliographystyle{IEEEtran}
\bibliography{IEEEabrv,horn_rt}

\begin{thebibliography}{10}
\providecommand{\url}[1]{#1}
\csname url@samestyle\endcsname
\providecommand{\newblock}{\relax}
\providecommand{\bibinfo}[2]{#2}
\providecommand{\BIBentrySTDinterwordspacing}{\spaceskip=0pt\relax}
\providecommand{\BIBentryALTinterwordstretchfactor}{4}
\providecommand{\BIBentryALTinterwordspacing}{\spaceskip=\fontdimen2\font plus
\BIBentryALTinterwordstretchfactor\fontdimen3\font minus
  \fontdimen4\font\relax}
\providecommand{\BIBforeignlanguage}[2]{{%
\expandafter\ifx\csname l@#1\endcsname\relax
\typeout{** WARNING: IEEEtran.bst: No hyphenation pattern has been}%
\typeout{** loaded for the language `#1'. Using the pattern for}%
\typeout{** the default language instead.}%
\else
\language=\csname l@#1\endcsname
\fi
#2}}
\providecommand{\BIBdecl}{\relax}
\BIBdecl

\bibitem{Fagin74SIAM}
R.~Fagin, ``Generalized first-order spectra and polynomial-time recognizable
  sets’,'' in \emph{Complexity of Computation, SIAM-AMS Proceedings}, 1974,
  pp. 43--73.

\bibitem{Libkin04FMT}
L.~Libkin, \emph{Elements of Finite Model Theory}, ser. Texts in Theoretical
  Computer Science. An {EATCS} Series.\hskip 1em plus 0.5em minus 0.4em\relax
  Springer, 2004.

\bibitem{Immerman99DesComp}
N.~Immerman, \emph{Descriptive complexity}.\hskip 1em plus 0.5em minus
  0.4em\relax Springer, 1999.

\bibitem{Gradel92TCS}
E.~Gr{\"{a}}del, ``Capturing complexity classes by fragments of second-order
  logic,'' \emph{Theoretical Computer Science}, vol. 101, no.~1, pp. 35--57,
  1992.

\bibitem{Gradel07EATCS}
E.~Gr{\"{a}}del, P.~G. Kolaitis, L.~Libkin, M.~Marx, J.~Spencer, M.~Y. Vardi,
  Y.~Venema, and S.~Weinstein, \emph{Finite Model Theory and Its
  Applications}.\hskip 1em plus 0.5em minus 0.4em\relax Springer, 2007.

\bibitem{Fischer65JACM}
P.~C. Fischer, ``Generation of primes by one-dimensional real-time iterative
  array,'' \emph{Journal of the ACM}, vol.~12, pp. 388--394, 1965.

\bibitem{Delorme02SCA}
M.~Delorme and J.~Mazoyer, ``Signals on cellular automata,'' in
  \emph{Collision-based Computing}, A.~Adamatzky, Ed.\hskip 1em plus 0.5em
  minus 0.4em\relax Springer, 2002, pp. 231--275.

\bibitem{BacqueyGO17ICALP}
N.~Bacquey, E.~Grandjean, and F.~Olive, ``{Definability by Horn Formulas and
  Linear Time on Cellular Automata},'' in \emph{ICALP 2017}, vol.~80, 2017, pp.
  99:1--99:14.

\bibitem{GrandjeanO16}
E.~Grandjean and F.~Olive, ``A logical approach to locality in pictures
  languages,'' \emph{Journal of Computer and System Science}, vol.~82, no.~6,
  pp. 959--1006, 2016.

\bibitem{AbiteboulHV95}
S.~Abiteboul, R.~Hull, and V.~Vianu, \emph{Foundations of Databases}.\hskip 1em
  plus 0.5em minus 0.4em\relax Addison-Wesley, 1995.

\bibitem{Cole69rtIA}
S.~N. Cole, ``Real-time computation by n-dimensional iterative arrays of
  finite-state machine,'' \emph{IEEE Transactions on Computing}, vol.~18, pp.
  349--365, 1969.

\bibitem{Smith72JCSS}
A.~R. Smith, ``Real-time language recognition by one-dimensional cellular
  automata,'' \emph{Journal of Computer and System Science}, vol.~6, pp.
  233--253, 1972.

\bibitem{Dyer80IC}
C.~R. Dyer, ``One-way bounded cellular automata,'' \emph{Information and
  Control}, vol.~44, no.~3, pp. 261--281, Mar. 1980.

\bibitem{Terrier99vNvsM}
V.~Terrier, ``Two-dimensional cellular automata recognizer,'' \emph{Theoretical
  Computer Science}, vol. 218, no.~2, pp. 325--346, May 1999.

\bibitem{Anael17DLT}
A.~Grandjean, ``Differences between 2d neighborhoods according to real time
  computation,'' in \emph{Developments in Language Theory}, 2017, pp. 198--209.

\bibitem{ChoffCulik84AI}
C.~Choffrut and K.~{\v{C}ul\'{\i}k II}, ``On real-time cellular automata and
  trellis automata,'' \emph{Acta Informatica}, vol.~21, no.~4, pp. 393--407,
  Nov. 1984.

\bibitem{IbarraJiang87oneWayCA}
O.~H. Ibarra and T.~Jiang, ``On one-way cellular arrays,'' \emph{SIAM Journal
  on Computing}, vol.~16, no.~6, pp. 1135--1154, Dec. 1987.

\bibitem{Poupet05STACS}
V.~Poupet, ``Cellular automata: Real-time equivalence between one-dimensional
  neighborhoods,'' in \emph{{STACS} 2005}, 2005, pp. 133--144.

\bibitem{CulikGS84SysTA}
K.~{\v{C}ul\'{\i}k II}, J.~Gruska, and A.~Salomaa, ``Systolic trellis automata.
  {I},'' \emph{International Journal Computer Mathematics}, vol.~15, no. 3-4,
  pp. 195--212, 1984.

\bibitem{Terrier96uvu}
V.~Terrier, ``Language not recognizable in real time by one-way cellular
  automata,'' \emph{Theoretical Computer Science}, vol. 156, no. 1--2, pp.
  281--287, Mar. 1996.

\bibitem{Okhotin04rairo}
A.~Okhotin, ``On the equivalence of linear conjunctive grammars and trellis
  automata,'' \emph{Theoretical Informatics and Applications}, vol.~38, no.~1,
  pp. 69--88, 2004.

\bibitem{Okhotin13CSR}
------, ``Conjunctive and boolean grammars: The true general case of the
  context-free grammars,'' \emph{Computer Science Review}, vol.~9, pp. 27--59,
  2013.

\bibitem{Terrier03b}
V.~Terrier, ``Characterization of real time iterative array by alternating
  device,'' \emph{Theoretical Computer Science}, vol. 290, no.~3, pp.
  2075--2084, 2003.

\end{thebibliography}

\bigskip
\begin{center}
{\bf APPENDICES}
\end{center}

\subsection*{Appendix A: Expressing natural problems in our logics}

\subsubsection*{Expressing problems $\mathtt{Palindrome}$ and $\mathtt{notPalindrome}$ in inclusion logic}

The problem $\mathtt{Palindrome}$ is expressed by the conjunction of the following clauses:

$x<y \land Q_s(x) \land Q_t(y) \to \mathtt{notPal}(x,y)$, for all $s,t\in\Sigma$ with $s\neq t$;

$x<y \land \mathtt{notPal}(x+1,y-1) \to \mathtt{notPal}(x,y)$;

$x\le y \land \mathtt{min}(x)  \land \mathtt{max}(y) \land \mathtt{notPal}(x,y)\to \bot$.

\medskip
\noindent
The problem $\mathtt{notPalindrome}$ is expressed by the conjunction of the following clauses:
$x=y \to \mathtt{Pal}(x,y)$, 

\noindent
$x\le y \land \mathtt{Successor}(x,y) \land Q_s(x) \land Q_s(y) \to \mathtt{Pal}(x,y)$,

\noindent
$x< y \land \mathtt{Pal}(x+1,y-1) \land Q_s(x) \land Q_s(y) \to \mathtt{Pal}(x,y)$, for all $s\in\Sigma$,
and $x\le y \land \mathtt{Pal}(x,y) \to \bot$.

\medskip
\noindent
Here, $\mathtt{Successor}(x,y)$ intuitively means $x+1=y$ and is defined by the following clauses: 
$x=y \to \mathtt{Equal}(x,y)$
and $x<y \land \mathtt{Equal}(x+1,y)\to \mathtt{Successor}(x,y)$.

\bigskip
\subsection*{Appendix B: Complements of the proofs of Lemmas~\ref{lemmapred} and~\ref{lemmaincl}}

\noindent
{\bf Step~10 of Lemma~\ref{lemmapred}: Elimination of atoms $R(x,y)$ as hypotheses:} 
The first idea is to group together in each computation clause the hypothesis atoms of the form $R(x,y)$ and the conclusion of the clause. Accordingly, the formula obtained $\Phi$ can be rewritten in the form\\
$\Phi:= \exists \mathbf{R} \forall x \forall y [ \bigwedge_{i} C_i (x,y) \land \bigwedge_{i\in[1,k]} (\alpha_i(x,y)\to \theta_i(x,y)) ]$,\\
where the $C_i$'s are the input clauses and the contradiction clause and each computation clause is written in the form $\alpha_i(x,y) \to \theta_i(x,y)$ where
$\alpha_i(x,y)$ is a conjunction of formulas of the only forms $R(x-1,y)\land \lnot\mathtt{min}(x)$, $R(x,y-1)\land \lnot\mathtt{min}(y)$, \emph{but not} $R(x,y)$,
and $\theta_i(x,y)$ is a Horn clause \emph{all} the atoms of which are of the form $R(x,y)$.

We number $R_1,\ldots,R_m$ the computation predicates of $\mathbf{R}$.
To each subset $J\subseteq [1,k]$ of the family of implications 
$(\alpha_i(x,y)\to \theta_i(x,y))_{i\in[1,k]}$ let us associate the set
\begin{center}
$K_J := \{h \in [1,m] \; \vert \; \bigwedge_{i\in J} \theta_i (x,y) \to R_h(x,y)$ is a tautology$\}$.
\end{center}

\noindent
Note that the notion of \emph{tautology} used in the definition of $K_J$ is purely ``propositional'' because all the atoms involved are of the form $R_i(x,y)$, i.e., refer to the same pair of variables ~$(x,y)$. Also, note that the function $J\mapsto K_J$ is \emph{monotonous}: for $J'\subseteq J$, we have $K_{J'}\subseteq K_J$ because $\bigwedge_{i\in J'} \theta_i (x,y) \to R_h(x,y)$ implies $\bigwedge_{i\in J} \theta_i (x,y) \to R_h(x,y)$.

\medskip
Clearly, it is enough to prove the following claim:

\begin{claim}\label{claim:horn} 
The formula $\Phi$ is equivalent to the \emph{normalized} formula
\begin{eqnarray*}
\lefteqn{\Phi':=\exists \mathbf{R} \forall x \forall y \; [\;\bigwedge_{i} C_i (x,y)}\\
&&\land  \bigwedge_{J\subseteq [1,k]} \bigwedge_{h\in K_J} 
(\bigwedge_{i\in J} \alpha_i(x,y)\to R_h(x,y)) \;].
\end{eqnarray*}
\end{claim}


\medskip
\noindent
 \emph{Proof of the implication} $\Phi \Rightarrow \Phi'$: It is enough to prove the implication  
\[
[\bigwedge_{i\in[1,k]}(\alpha_i(x,y) \to \theta_i(x,y))]\to [\bigwedge_{i\in J} \alpha_i(x,y)\to \bigwedge_{h\in K_J}R_h(x,y)]
\]
for all 
set $J\subseteq [1,k]$. 
The implication to be proved can be equivalently written:
\[
[\bigwedge_{i\in J} \alpha_i(x,y) \land \bigwedge_{i\in[1,k]} (\alpha_i(x,y) \to \theta_i(x,y))]
\to \bigwedge_{h\in K_J} R_h(x,y).
\]
This implication is a straightforward consequence of the two following facts: 
The sub-formula between brackets above implies the conjunction $\bigwedge_{i\in J} \theta_i(x,y)$. 
As the implication $\bigwedge_{i\in J} \theta_i (x,y) \to \bigwedge_{h\in K_J} R_h(x,y)$ 
is a tautology (by definition of $K_J$), the implication to be proved is a tautology too. 

\medskip
The converse implication $\Phi' \Rightarrow \Phi$ is more difficult to prove. It uses a folklore property of propositional Horn formulas easy to be proved:

\begin{lemma}[Horn property: folklore]\label{lemma:folklore}
Let $F$ be a \emph{strict Horn formula} of propositional calculus, that is a conjunction of clauses of the form $p_1\land\ldots\land p_k\to p_0$ where $k\ge 0$ and the $p_i$'s are propositional variables. Let $F'$ be the conjunction of propositional variables $q$ such that the implication $F\to q$ is a tautology. $F$ has the same minimal model\footnote{For example,
for $F:=p_1\land p_3\land (p_1\land p_3\to p_5) \land (p_1\land p_2 \to p_4)$, we have $F':=p_1\land p_3\land p_5$ which has the same minimal model $I$ as $F$; this model is given by $I(p_1)=I(p_3)=I(p_5)=1$ and $I(p_2)=I(p_4)=0$.} as $F'$.
\end{lemma}

\smallskip
\noindent
\emph{Proof of the implication} $\Phi' \Rightarrow \Phi$: 
Let~$\langle w \rangle$ be a model of $\Phi'$ and let $(\langle w \rangle,\mathbf{R})$ be the minimal model of the Horn formula
\begin{eqnarray*}
\lefteqn{\varphi':=\forall x \forall y \; [\;\bigwedge_{i} C_i (x,y)}\\
&&\land  \bigwedge_{J\subseteq [1,k]} \bigwedge_{h\in K_J} 
(\bigwedge_{i\in J} \alpha_i(x,y)\to R_h(x,y))\;].
\end{eqnarray*}

\noindent
It is enough to show that $(\langle w \rangle,\mathbf{R})$ also satisfies the formula
\begin{eqnarray*}
\lefteqn{\varphi:=\forall x \forall y \; [\;\bigwedge_{i} C_i (x,y)}\\ 
&& \land \bigwedge_{i\in[1,k]} (\alpha_i(x,y)\to \theta_i(x,y))\;].
\end{eqnarray*}

\noindent
As each $\alpha_i$ is a conjunction of formulas of the form\linebreak
$R(x-1,y)\land\lnot \mathtt{min}(x)$, or $R(x,y-1)\land\lnot \mathtt{min}(y)$, 
we make an induction on the domain $\{(a,b)\in[1,n]^2\;\vert\; a+b\le t\}$, for $t\in[1,2n]$. 
More precisely, we are going to prove, by recurrence on the integer $t\in[1,2n]$, that the minimal model $(\langle w \rangle,\mathbf{R})$ of $\varphi'$ satisfies the ``relativized''  formula $\varphi_t$ of the formula $\varphi$ defined by
\begin{eqnarray*}
\lefteqn{\varphi_t:=\forall x \forall y \;
[\;x+y\le t \to }\\ 
&&[\;\bigwedge_{i} C_i (x,y) \land \bigwedge_{i\in[1,k]} (\alpha_i(x,y)\to \theta_i(x,y))\;]
\;]
\end{eqnarray*}
As the hypothesis $x+y\le 2n$ holds for all $x,y$ in the domain $[1,n]$, $\varphi_{2n}$ is equivalent to $\varphi$ on the structure $(\langle w \rangle,\mathbf{R})$.

\medskip
\emph{Basis case:} 
For $t=1$ the set $\{(a,b)\in[1,n]^2\;\vert\; a+b\le t\}$ is empty so that the ``relativized''  formula $\varphi_1$ is trivially true in the minimal model $(\langle w \rangle,\mathbf{R})$ of $\varphi'$.

\medskip
\emph{Recurrence step:} Suppose $(\langle w \rangle,\mathbf{R})\models\varphi_{t-1}$, for an integer $t\in[2, 2n]$.
It is enough to show that, for each couple $(a,b)\in [1,n]^2$ such that $a+b=t$, we have 
$ (\langle w \rangle,\mathbf{R})\models \bigwedge_{i\in[1,k]} (\alpha_i(a,b)\to \theta_i(a,b))$. Let $J_{a,b}$ be the set of indices $i\in [1,k]$ such that the couple $(a,b)$ satisfies~$\alpha_i$:
\[
J_{a,b}:=\{ i\in [1,k] \;\vert\;  (\langle w \rangle,\mathbf{R})\models \alpha_i(a,b)\}. 
\]
Recall that each $\alpha_i(a,b)$ is a (possibly empty) conjunction of atoms $R(a',b')$ with $(a',b')=(a-1,b)$ or $(a',b')=(a,b-1)$, therefore such that $a'+b'=t-1$. 
Let any set $J\subseteq [1,k]$. Let us examine the two possible cases:

\medskip
1) $J\subseteq J_{a,b}$: then the conjunction $\bigwedge_{i\in J}\alpha_i(a,b)$ holds in $(\langle w \rangle,\mathbf{R})$; hence, in $(\langle w \rangle,\mathbf{R})$, the conjunction 
$\bigwedge_{h\in K_J} (\bigwedge_{i\in J} \alpha_i(a,b)\to R_h(a,b))$ is equivalent to
$\bigwedge_{h\in K_J} R_h(a,b)$;

\smallskip
2) $J\setminus J_{a,b}\neq \emptyset$: then the conjunction $\bigwedge_{i\in J}\alpha_i(a,b)$ is false in $(\langle w \rangle,\mathbf{R})$; 
hence, the conjunction $\bigwedge_{h\in K_J} (\bigwedge_{i\in J} \alpha_i(a,b)\to R_h(a,b))$ holds in~$(\langle w \rangle,\mathbf{R})$.

\medskip
\noindent
From (1) and (2), we deduce that in $(\langle w \rangle,\mathbf{R})$ the conjunction
$
\bigwedge_{J\subseteq [1,k]} \bigwedge_{h\in K_J} (\bigwedge_{i\in J} \alpha_i(a,b)\to R_h(a,b))
$
is equivalent to the conjunction 
$\bigwedge_{J\subseteq J_{a,b}} \bigwedge_{h\in K_J}  R_h(a,b)$, which can be simplified as 
$\bigwedge_{h\in K_{J_{a,b}}} R_h(a,b)$
because $J\subseteq J_{a,b}$ implies $K_J\subseteq K_{J_{a,b}}$.
Consequently, for all $h\in[1,m]$, the minimal model $(\langle w \rangle,\mathbf{R})$ of the Horn formula
$\varphi'$ satisfies the atom $R_h(a,b)$ \emph{iff} $h$ belongs to $K_{J_{a,b}}$. 
By definition,
\begin{center}
$K_{J_{a,b}} := \{h \in [1,m] \; \vert \; \bigwedge_{i\in J_{a,b}} \theta_i (x,y) \to R_h(x,y)$ is a tautology$\}$
\end{center}
or, equivalently,
\begin{center}
$K_{J_{a,b}} := \{h \in [1,m] \; \vert \; \bigwedge_{i\in J_{a,b}} \theta_i (a,b) \to R_h(a,b)$ is a tautology$\}$.
\end{center}
As a consequence of Lemma~\ref{lemma:folklore}, the two conjunctions $\bigwedge_{i\in J_{a,b}} \theta_i (a,b)$
and $\bigwedge_{h\in K_{J_{a,b} }}R_h(a,b)$ have the same minimal model, which is also the restriction of the minimal model $(\langle w \rangle,\mathbf{R})$ of $\varphi'$ to the set of atoms $R_h(a,b)$, for $h\in[1,m]$.
Therefore, if $i\in J_{a,b}$, then $(\langle w \rangle,\mathbf{R})\models \theta_i (a,b)$.
If  $i\in [1,k] \setminus J_{a,b}$, then we have
$(\langle w \rangle,\mathbf{R})\models \lnot \alpha_i(a,b)$, by definition of $J_{a,b}$.
Therefore, for all $i\in [1,k]$, we get 
$(\langle w \rangle,\mathbf{R})\models \lnot \alpha_i(a,b)\lor \theta_i (a,b)$.
In other words, for all $(a,b)$ such that $a+b=t$:
\[
(\langle w \rangle,\mathbf{R})\models \bigwedge_{i\in[1,k]}(\alpha_i(a,b) \to \theta_i (a,b))
\]
and then $(\langle w \rangle,\mathbf{R})\models \varphi_t$.

This concludes the inductive proof that $(\langle w \rangle,\mathbf{R})\models \varphi_t$, for all $t\in [1,2n]$, and then $\langle w \rangle \models \Phi$. This proves the converse implication $\Phi' \Rightarrow \Phi$. Claim~\ref{claim:horn} is demonstrated.
$\square$

\medskip
\noindent
{\bf General case of Lemma~\ref{lemmapred}:} Steps~1-6 are easy to adapt in the general case where the initial formula may contain hypotheses of the form $R(y-b,x-a)$. The new points are the following:
Step~3 restricts the computation atoms to four forms:  $R(x,y)$, $R(y,x)$, $R(x-1,y)$ and $R(x,y-1)$; the key point is the adaptation of step~7 (folding the domain) so that it eliminates the atoms of the form $R(y,x)$. 
Without loss of generality, assume that each computation clause is of one of the following forms:

(a) $S(x-1,y)\land  \lnot \mathtt{min}(x) \to R(x,y)$;

(b) $S(x,y-1)\land  \lnot \mathtt{min}(y) \to R(x,y)$;

(c) $S(x,y)\land T(x,y) \to R(x,y)$;
(d) $S(y,x)\to R(x,y)$.\\
Here again, this is obtained by ``decomposing'' each computation clause into an ``equivalent'' conjunction of clauses using new intermediate predicates. For instance, the computation clause
$R_1(x-1,y)\land  \lnot \mathtt{min}(x)\land R_2(x,y-1)\linebreak \land \lnot \mathtt{min}(y)\land R_3(y,x)$
$\to R_4(x,y)$
is ``equivalent'' to the conjunction of the following clauses using the new predicates $R_5,R_6,R_7,R_8$:
$R_1(x-1,y)\land  \lnot \mathtt{min}(x)\to R_5(x,y)$;\\
$R_2(x,y-1)\land  \lnot \mathtt{min}(y)\to R_6(x,y)$;
$R_5(x,y)\land R_6(x,y)\to R_7(x,y)$;
$R_3(y,x)\to R_8(x,y)$;
$R_7(x,y)\land R_8(x,y) \to R_4(x,y)$.

The folding of clauses (a-c) is not modified. 
Let us describe how to fold the (new) clauses (d): $S(y,x)\to R(x,y)$.\\
Obviously, such a clause is equivalent to the conjunction of the two clauses
(i) $x\le y \land S(y,x)\to R(x,y)$ and \linebreak
(ii) $y\le x  \land S(y,x)\to R(x,y)$.
The equivalent ``folded''  form of clause (i) is $x\le y \land S^{\mathtt{inv}}(x,y)\to R(x,y)$.
The clause (ii) is equivalent to the clause
$x\le y \land S(x,y) \to R(y,x)$
the equivalent ``folded''  form of which is
$x\le y \land S(x,y) \to R^{\mathtt{inv}}(x,y)$.
Finally, steps 8-10 are not modified.

\bigskip
\noindent
{\bf Step 6 of Lemma~\ref{lemmaincl}: Elimination of atoms $Q_s(x-a)$, $Q_s(y+b)$, for $a,b>0$:}
This step is quite similar to step 5. For each $R\in\mathbf{R}$, we introduce new predicates:

\smallskip
\noindent
$R_{\hyp s_1,\dots,s_l,t_1, \dots, t_m}^{x-a_1,\dots,x-a_l,y+b_1,\dots,y+b_m}$ with $l,m\ge 0$, the $s_i,t_j\in \Sigma$, 

\smallskip
\noindent
$0\le a_1<a_2\ldots<a_l\le A$ and $0\le b_1<b_2\ldots<b_m\le~B$, 
where $A$ (resp.~$B$) is the maximal $a$ in atoms $Q_s(x-a)$ 
(resp. maximal $b$ in atoms $Q_s(y+b)$).
Their intuitive meaning is as follows:

\smallskip
\noindent
$R_{\hyp s_1,\dots,s_l,t_1, \dots, t_m}^{x-a_1,\dots,x-a_l,y+b_1,\dots,y+b_m}(x,y) \iff$

\smallskip
$[\;[\bigwedge_{i=1,\ldots,l}(Q_{s_i}(x-a_i) \land x>a_i)$

$\land \bigwedge_{j=1,\ldots,m}(Q_{t_j}(y+b_j)\land y\le n-b_j)] \to R(x,y)\;]$.

\smallskip
 \emph{Transforming the initialization clauses:} Each initialization clause 
 $x=y \land Q_s(x-a) \land x>a \to R(x,y)$, $a\ge1$, 
 (resp. $x=y \land Q_s(y+b) \land y \leq n-b \to R(x,y)$, $b\ge1$)
 is transformed into the clause $x=y \land R_{\hyp, s}^{x-a}  \to R(x,y)$
 (resp.  $x=y \land R_{\hyp, s}^{y+b} \to R(x,y)$).
 
 \smallskip
 \emph{Transforming the computation clauses:} To each clause (a) $x<y \land S(x+1,y) \to R(x,y)$ add the following clauses justified by the identity $x+1-a_i=x-(a_i-1)$:

\smallskip 
\noindent
$x<y \land S_{\hyp s_1,\dots,s_l,t_1, \dots, t_m}^{x-a_1,\dots,x-a_l,y+b_1,\dots,y+b_m}(x+1,y)$

\smallskip 
$\to R_{\hyp s_1,\dots,s_l,t_1, \dots, t_m}^{x-(a_1-1),\dots,x-(a_l-1),y+b_1,\dots,y+b_m}(x,y)$.

\smallskip 
\noindent
Similarly, to each clause (b) $x<y \land S(x,y-1) \to R(x,y)$ add the clauses

\noindent
$x<y \land S_{\hyp s_1,\dots,s_l,t_1, \dots, t_m}^{x-a_1,\dots,x-a_l,y+b_1,\dots,y+b_m}(x,y-1)$

\smallskip 
$\to R_{\hyp s_1,\dots,s_l,t_1, \dots, t_m}^{x-a_1,\dots,x-a_l,y+(b_1-1),\dots,y+(b_m-1)}(x,y)$.

\smallskip 
\noindent
Moreover, add for $a_1=0$ and each $s_1\in\Sigma$, the following ``verification'' clauses, which intuitively delete the hypothesis
$Q_{s_1}(x)$ after verifying that it is satisfied because of the equivalence $W_{s_1}^x(x,y)\iff Q_{s_1}(x)$:

\smallskip 
\noindent
$x<y \land S_{\hyp s_1,s_2,\dots,s_l,t_1, \dots, t_m}^{x, x-a_2,\dots,x-a_l,y+b_1,\dots,y+b_m}(x,y) \land W_{s_1}^x(x,y)$

\smallskip 
$\to R_{\hyp s_2,\dots,s_l,t_1, \dots, t_m}^{x-a_2,\dots,x-a_l,y+b_1,\dots,y+b_m}(x,y)$.

\smallskip 
\noindent
Similarly, add for $b_1=0$ and each $t_1\in\Sigma$, the ``verification'' clauses (justified by $W_{t_1}^y(x,y) \iff Q_{t_1}(y)$) :

\smallskip 
\noindent
$x<y \land S_{\hyp s_1,\dots,s_l,t_1, t_2,\dots, t_m}^{x-a_1,\dots,x-a_l,y, y+b_2,\dots,y+b_m}(x,y) \land W_{t_1}^y(x,y)$

\smallskip 
$\to R_{\hyp s_1,\dots,s_l,t_2\dots, t_m}^{x-a_1,\dots,x-a_l,y+b_2,\dots,y+b_m}(x,y)$.

\smallskip 
\noindent
For each clause (c) $x \preceq y \land S(x,y) \land T(x,y) \to R(x,y)$, where~$\preceq\;\in \{<,=\}$, add similar clauses that \emph{cumulate} the hypotheses provided they are \emph{compatible}: for example, the clause 

\smallskip 
\noindent
$x\preceq y  \land S_{s_1,s_2,t_1}^{x-1,x-3,y+2}(x,y) \land T_{s_1,t_1,t_2}^{x-1,y+2,y+4}(x,y)$
 
 \smallskip 
  $\to R_{s_1,s_2,t_1,t_2}^{x-1,x-3,y+2,y+4}(x,y).$

\end{document}